\documentclass[a4paper,UKenglish]{article}

\usepackage{fullpage}
\usepackage{amsthm,amsmath,mathrsfs,amssymb}
\usepackage{graphics,graphicx}
\usepackage{epsfig}
\usepackage{comment}
\usepackage{color}
\usepackage{nccmath}
\usepackage{authblk}

\usepackage{environ}

\usepackage{array}
\usepackage{algorithm}
\usepackage{algpseudocode}
\usepackage[draft]{fixme}

\newtheorem{theorem}{Theorem}
\newtheorem{lemma}{Lemma}

\newtheorem{definition}{Definition}

\algrenewcommand{\algorithmicrequire}{\textbf{Input:}}
\algrenewcommand{\algorithmicensure}{\textbf{Output:}}

\usepackage{tikz}
\usetikzlibrary{patterns,snakes}
\tikzstyle{overbrace text style}=[font=\tiny, above, pos=.5, yshift=5pt]
\tikzstyle{overbrace style}=[decorate,decoration={brace,raise=5pt,amplitude=3pt}]

\usepackage[ligature]{semantic} 
\mathlig{->}{\rightarrow} \mathlig{-->}{\longrightarrow}
\mathlig{=>}{\Rightarrow} \mathlig{==>}{\Longrightarrow} 
\mathlig{<-}{\leftarrow} \mathlig{<--}{\longleftarrow}
\mathlig{<=}{\Leftarrow} \mathlig{<==}{\Longleftarrow}
\mathlig{<->}{\leftrightarrow} \mathlig{<-->}{\longleftrightarrow}
\mathlig{<=>}{\Leftrightarrow} \mathlig{<==>}{\Longleftrightarrow}
\mathlig{...}{\ldots} 
\mathlig{....}{\ldots} 
\mathlig{==}{\equiv}

\newcommand{\RR}{\mathbb{R}}

\newcommand{\xx}{\mathbf{x}}

\newcommand{\TT}{\mathcal{T}}
\newcommand{\CCC}{\mathcal{C}}

\newcommand{\CSH}{{\sc ConHalving}}
\newcommand{\GC}{{\sc Gcircuit}}
\newcommand{\tucker}{{\sc Tucker}}

\newtheorem{problem}{\textbf{Problem}}

\NewEnviron{problem1}[1]{%
	\begin{center}\fbox{\parbox{6in}{%
				{\centering\scshape #1\par}%
				\parskip=1ex
				\everypar{\hangindent=1em}%
				\BODY
			}}\end{center}}

			\title{Hardness Results for Consensus-Halving}
			
			\author[1]{Aris Filos-Ratsikas}
			\author[2]{S{\o}ren Kristoffer Stiil Frederiksen\\}
			\author[3]{Paul W. Goldberg}
			\author[4]{Jie Zhang}
			\affil[1]{\'{E}cole Polytechnique F\'{e}d\'{e}rale de Lausanne, Switzerland\\
				\texttt{aris.filosratsikas@epfl.ch}}
			\affil[2]{Aarhus University, Denmark\\
				\texttt{sorensf@gmail.com}}
			\affil[3]{University of Oxford, United Kingdom\\
				\texttt{paul.goldberg@cs.ox.ac.uk}}
			\affil[4]{University of Southampton, United Kingdom\\
				\texttt{jie.zhang@soton.ac.uk}}

	\date{}		
			
\begin{document}


\maketitle

\begin{abstract}
\noindent The \emph{Consensus-halving} problem is the problem of dividing an object into two portions, such that each of $n$ agents has equal valuation for the two portions. We study the $\epsilon$-approximate version, which allows each agent to have an $\epsilon$ discrepancy on the values of the portions. It was recently proven in \cite{filos2018consensus} that the problem of computing an $\epsilon$-approximate consensus-halving solution (for $n$ agents and $n$ cuts)  is PPA-complete when $\epsilon$ is \emph{inverse-exponential}. In this paper, we prove that when $\epsilon$ is \emph{constant}, the  problem is PPAD-hard and the problem remains PPAD-hard when we allow a constant number of additional cuts. Additionally, we prove that deciding whether a solution with $n-1$ cuts exists for the problem is NP-hard. 
\end{abstract}

\section{Introduction}

Suppose that two families want to split a piece of land into two regions such that every member of each family believes that the land is equally divided, or suppose that a conference organizer wants to assign the conference presentations to the morning and the afternoon sessions, so that every participant thinks that the two sessions are equally interesting. Is it possible to achieve these objectives? If yes, how can it be done and how efficiently? What if we aim for ``almost equal'' instead of ``equal''?

These real-life problems can be modeled as the \emph{Consensus-halving problem} 
\cite{simmons2003consensus}. More formally, there are $n$ agents and an object to be divided; each agent may have a different opinion as to which part of the object is more valuable. The problem is to divide the object into  two portions such that each of the $n$ agents believes that the two portions have equal value, according to her personal opinion. The division may need to cut the object into pieces and then label each piece appropriately to include it in one of the two portions.

The importance of the Consensus-halving problem - or to be precise, of its approximate version, where there is an associated precision parameter $\epsilon$ - other than the fact that it models real-life problems like the ones described above, lies in in the following fact: It is the first ``natural'' problem that is \emph{complete} for the complexity class PPA, where ``natural'' here means that its does not contain a circuit explicitly in its definition; this was proven quite recently by Filos-Ratsikas and Goldberg \cite{filos2018consensus}. PPA is a class of \emph{total search problems} \cite{megiddo1991total} defined in \cite{papadimitriou1994complexity}, and is a superclass of the class PPAD, which precisely captures the complexity of computing a Nash equilibrium \cite{daskalakis2009complexity,chen2009settling}. Therefore, generally speaking, a PPA-hardness result is stronger than a PPAD-hardness result.

Crucially however, the hardness result in \cite{filos2018consensus} requires the precision parameter to be inverse-exponential in the number of agents and does not even provably preclude the possibility of efficient algorithms, if we allow larger discrepancies in the values for the two portions. In this paper, we prove that this is actually not possible\footnote{Under usual computational complexity assumptions, here that PPAD-hard problems do not admit polynomial-time algorithms.}, by showing that even when the allowed discrepancy is independent of the number of agents, the problem is PPAD-hard. Understanding the problem for increasing values of the discrepancy parameter is quite important in terms of capturing precisely its complexity and resembles closely the series of results establishing hardness of computing a mixed $\epsilon$-Nash equilibrium, from $\epsilon$ being inverse-polynomial in \cite{daskalakis2009complexity,chen2009settling} to being constant in \cite{rubinstein2015inapproximability}, as well as several other problems (see \cite{rubinstein2015inapproximability}). Additionally, one could imagine that solutions where constant discrepancies are acceptable are the ones arising in several real-life scenarios, such as splitting land.

\subsection{Our results}

We are interested in the computational complexity of computing an $\epsilon$-approximate solution to the Consensus-halving problem where $\epsilon$ is a constant function of the number of agents, as well as the complexity of deciding whether given an input instance, $n-1$ cuts are sufficient to achieve an $\epsilon$-approximate solution. We discuss our main results below.
\medskip
\begin{itemize}
\item We prove that the problem of finding an $\epsilon$-approximate solution to the Consensus-halving problem for $n$ agents using $n$ cuts is PPAD-hard. Moreover, the problem remains PPAD-hard even if we allow a constant number of additional cuts. The result is established via a reduction from the approximate Generalized Circuit problem~\cite{chen2006settling,daskalakis2009complexity,rubinstein2015inapproximability}. 
\medskip
\item We prove that it is NP-hard to decide whether or not an $\epsilon$-approximate solution to the Consensus-halving problem for $n$ agents using $n-1$ cuts exists. We establish the result via a reduction from {\sc 3SAT}. 
\medskip
\item We prove that the problem of finding an $\epsilon$-approximate solution to the Consensus-halving problem for $n$ agents using $n$ cuts is in the computational class PPA; we obtain the result via a reduction to the computational version of Tucker's Lemma \cite{papadimitriou1994complexity,aisenberg20152}. 
\end{itemize}

\noindent We remark here that an earlier version of this paper actually predated  \cite{filos2018consensus} and some of the results in \cite{filos2018consensus} are established by referencing the results in the present paper. Specifically:
\medskip
\begin{itemize}
	\item While the authors of \cite{filos2018consensus} provide a rather elaborate reduction to establish PPA-hardness of the problem, the containment in the class PPA is established with reference to the present paper. In turn, the containment result follows from a formalization of the ideas of the algorithms by \cite{simmons2003consensus} and \cite{prescott2005constructive} and Fan's version of Tucker's Lemma \cite{fan1967simplicial,tucker1945some}.
	\medskip
	\item In \cite{filos2018consensus}, the authors obtain a \emph{computational equivalence} between the \emph{Necklace Splitting problem} \cite{alon1987splitting} and $\epsilon$-Consensus Halving, for $\epsilon$ being at least inverse-polynomial. The inverse-polynomial dependence on $\epsilon$ implies that PPA-hardness of the former problem does not follow from their hardness result, but PPAD-hardness does follow from their reduction and our main result here.
\end{itemize}


\subsection{Related work} 

The Consensus-Halving problem was explicitly formalized and studied firstly by Simmons and Su \cite{simmons2003consensus}, who proved that a solution with $n$ cuts always exists
and constructed a protocol that finds an approximate solution, which allows for a small discrepancy on the values of the two portions. Their proofs are based on one of the most applied theorems in topology, the Borsuk-Ulam Theorem \cite{borsuk1933drei} and its combinatorial analogue, known as Tucker's Lemma \cite{tucker1945some}. The existence of solutions to the problem was already known since \cite{goldberg1985bisection,alon1986borsuk,barbanel1996super} but the algorithm in \cite{simmons2003consensus} is constructive, in the sense that it actually finds such a solution and furthermore, it does not require the valuations of the players to be additively separable over subintervals, like some of the previous papers do. Actually, for the case of valuations which are probability measures, the existence of a solution with $n$ cuts was known since as early as the 1940s \cite{neyman1946theoreme} and can also be obtained as an application of the Hobby-Rice Theorem \cite{hobby1965moment} (also see \cite{alon1987splitting}). Despite proposing an explicit protocol however, the authors in \cite{simmons2003consensus} do not answer the question of ``efficiency'', i.e. how fast can a protocol find an (approximate) solution and the running time of their protocol is worst-case exponential-time.\footnote{The protocol exhaustively iterates through all the vertices of triangulated geometric object, which, to achieve a small discrepancy, has to be exponentially large.} 

To this end, Filos-Ratsikas and Golberg \cite{filos2018consensus} recently proved that the problem is PPA-complete, but as we explained in the introduction, the hardness holds only when the precision parameter is exponentially-small and therefore does not subsume our results. The computational classes PPA (Polynomial Parity Arguments) and PPAD (Polynomial Parity Arguments on Directed graphs) were introduced by Papadimitriou \cite{papadimitriou1994complexity} in an attempt to capture the precise complexity of several interesting problems of a topological nature  such as computational analogues of Sperner's Lemma \cite{sperner1928neuer} and Brouwer's and Kakutani's fixed point theorems \cite{border1989fixed}, which are all known to be in PPAD \cite{papadimitriou1994complexity}. Interestingly, Aisenberg et al. \cite{aisenberg20152} recently proved that the search problems associated with the Borsuk-Ulam Theorem and Tucker's Lemma are PPA-complete; this is the starting point for the reduction in \cite{filos2018consensus}, but it will also be used for our ``in-PPA'' result, which complements the hardness result of \cite{filos2018consensus}.

Our PPAD-hardness reduction goes via the {\em Generalized Circuit} problem.  Generalized circuits differ from usual circuits in the sense that they can contain cycles, a fact which basically induces a continuous function on the gates, and the solution is guaranteed to exist by Brouwer's fixed point theorem. The $\epsilon$-approximate Generalized Circuit problem was implicitly proven to be PPAD-complete for exponentially small $\epsilon$ in \cite{daskalakis2009complexity} and explicitly for polynomial small $\epsilon$ \cite{chen2006settling}, en route to proving that perhaps the most interesting problem in PPAD, that of computing a mixed-Nash equilibrium, is also complete for the class. The same problem was also used in \cite{chen2013complexity} to prove that finding an approximate market equilibrium for the Arrow-Debreu market with linear and non-monotone utilities is PPAD-complete and in \cite{othman2014complexity} to prove that finding an approximate solution of the Competitive Equilibrium with Equal Incomes (CEEI) for indivisible items is PPAD-complete. More recently,  Rubinstein \cite{rubinstein2015inapproximability} showed that computing an $\epsilon$-approximate solution for the Generalized Circuit problem is PPAD-complete for a constant $\epsilon$, which implies that finding an $\epsilon$-approximate Nash equilibrium is PPAD-complete for constant $\epsilon$, in the context of graphical games; we reduce from that version of the problem. This improvement should also lead to stronger hardness results in \cite{chen2013complexity} and \cite{othman2014complexity}, as well as other problems that rely on the Generalized Circuit problem. 

The Consensus-Halving problem is a typical fair division problem that studies how to divide a set of resources between a set of agents who have valuations on the resources, such that some fairness properties are fulfilled. The fair division literature, which dates back to the late 1940s \cite{steinhaus1948problem}, has studied a plethora of such problems, with \emph{chore-division} \cite{peterson2002four,gardner1978aha}, \emph{rent-partitioning} \cite{haake2002bidding,brams2001competitive,su1999rental} and the perhaps the most well-known one, \emph{cake-cutting} \cite{brams1996fair,robertson1998cake} being notable examples. Note that Consensus-halving is inherently different from (envy-free or proportional) cake-cutting, since the objective is that all participants are (approximately) equally satisfied with the solution, and they do not have ``ownership'' over the resulting pieces.

\section{Preliminaries}\label{sec:preliminaries}

We represent the object $O$ as a line segment $[0,1]$. Each agent in the set of agents $N=\{1,\ldots,n\}$ has its own valuation over any subset of interval $[0,1]$. These valuations are:
\smallskip
\begin{itemize}
	\item \emph{non-negative and bounded}, i.e. there exists $M>0$, such that for any subinterval $X \subseteq [0,1]$, it holds that  $0 \le u_i(X) \le M$.
	\item \emph{non-atomic}, i.e. agents have no value on any single point on the interval.
\end{itemize}  
\smallskip
For simplicity, the reader may assume that the valuations are represented as step functions (where agents have constant values over distinct intervals, although this is not necessary for the results to hold.\footnote{The containment result actually holds for more general functions, while our hardness results (PPAD-hardness and NP-hardness) hold even for well-behaved functions, such as step functions. We note here that while an exact solution to Consensus-Halving generally requires the valuations to be continuous, this is not necessary for the existence of an approximate solution; the algorithm of \cite{simmons2003consensus} can find such a solution assuming that valuations are bounded and non-atomic.} 
A set of $k$ cuts $\{t_1,\ldots,t_k\}$, where $0 \le t_1 \le \ldots \le t_k \le 1$, means that we cut along the points $t_1, \ldots, t_k$, such that the object is cut into $k+1$ pieces $X_i=[t_{i-1}, t_{i}]$ for $i=1,\ldots,k+1$, where $t_0=0$ and $t_{k+1}=1$. A labelling of an interval $X_i$ means that we assign a label $\ell \in \{+,-\}$ to $X_i$, which corresponds to including $X_i$ in a set of intervals, either $O_{+}$ or $O_{-}$. In case some cuts happen to be on the same point, say $t_{j-1}=t_{j}$, then the corresponding subinterval $X_j$ is a single point on which agents have no value. We will consider cuts on the same points to be the same cut, e.g. if there is only one such occurrence, we will consider the number of cuts to be $k-1$.

The Consensus-halving problem is to divide the object $O$ into two portions $O_{+}$ and $O_{-}$, such that every agent derives equal valuation from the two portions, i.e., $u_i(O_{+})=u_i(O_{-}), \forall i\in N$. The $\epsilon$-approximate Consensus-halving problem allows that each agent has a small discrepancy on the values of the two partitions, i.e., $|u_i(O_{+})-u_i(O_{-})|<\epsilon$. As discussed in the Introduction, such as solution always exists. \cite{simmons2003consensus}.

\medskip

\noindent We define the following search problem, called $(n,k,\epsilon)$-\CSH. 


\begin{problem}{$(n,k,\epsilon)$-\CSH}
	
	\textbf{Input:} The value density functions (valuation functions) $v_i : O -> R_+, i=1,\cdots,n$.
	
	\textbf{Output:} A partition $(O_{+},O_{-})$ with $k$ cuts such that $|u_i(O_{+})-u_i(O_{-})|\leq\epsilon$.
\end{problem}	
We will also consider the following decision problem, called $(n,n-1,\epsilon)$-\CSH. Note that for $n$ agents and $n-1$ cuts, a solution to $\epsilon$-approximate Consensus-halving problem is not guaranteed to exist.

\begin{problem}{$(n,n-1,\epsilon)$-\CSH}
	
	\textbf{Input:} The value density functions $v_i : O -> R_+, i=1,\cdots,n$.
	
	\textbf{Output:} {\sc Yes}, if a partition $(O_{+},O_{-})$ with $n-1$ cuts such that $|u_i(O_{+})-u_i(O_{-})|\leq\epsilon$ for all agents $i \in N$ exists, and {\sc No} otherwise.
\end{problem}	


\noindent \emph{TFNP, PPA and PPAD:} Most of the problems that we will consider in this paper belong to the class of \emph{total search problems}, i.e. search problems for which a solution is guaranteed to exist, regardless of the input. In particular, we will be interested in problems in the class TFNP, i.e. total search problems for which a solution is verifiable in polynomial time \cite{megiddo1991total}. 
An important subclass of TFNP is the class PPAD, defined by Papadimitriou in \cite{papadimitriou1994complexity}. PPAD stands for ``Polynomial Parity Argument on a Directed graph'' and is defined formally in terms of the problem {\sc End-Of-Line} \cite{papadimitriou1994complexity}. 
The class PPAD is defined in terms of an exponentially large digraph $G=(V,E)$ consisting of $2^n$ vertices with indegree and outdegree at most $1$. An edge between vertices $v_1$ and $v_2$ is present in $E$ if and only if the successor $S(v_1)$ of node $v_1$ is $v_2$ and the predecessor $P(v_2)$ of node $v_2$ is $v_1$. By construction, the point $0^n$ has indegree $0$ and we are looking for a point with outdegree $0$, which is guaranteed to exist. Note that the graph is given \emph{implicitly} to the input, through a function that given any vertex $v$, returns its set of neighbours (predecessor and successor) in polynomial time. PPAD is a subclass of the class PPA (``Polynomial Parity Argument'') which is defined similarly, but in terms of an undirected graph in which every vertex has degree at most $2$, and given a vertex of degree $1$, we are asked to find another vertex of degree $1$; the computational problem associated with the class is called {\sc Leaf} \cite{papadimitriou1994complexity} and a problem is the class PPA if it is polynomial-time reducible to {\sc Leaf}.

For clarity of exposition, we postpone the formal definitions of {\sc End-Of-Line} and {\sc Leaf} until Section \ref{sec:inPPA}, where they are being used for the ``in-PPA'' result.

\subsection{Generalized Circuits}

A {\em generalized circuit} $S=(V, \TT)$ consists of a set of nodes $V$ and a set of gates $\TT$ and let $N=|V|$ and $M=|\TT|$. Every gate $T \in \TT$ is a 5-tuple $T =(G, v_{in_1}, v_{in_2}, v_{out}, \alpha)$ where
\begin{itemize}
\item $G\in \{ G_{\zeta}, G_{\times \zeta}, G_{=}, G_{+}, G_{-}, G_{<}, G_{\vee}, G_{\wedge},  G_{\neg} \}$ is the type of the gate, 
\item $v_{in_1}, v_{in_2} \in V \cup \{nil\}$ are the first and second input nodes of the gate or $nil$ if not applicable,
\item $v_{out}\in V$ is the output node, and $\alpha \in [0,1] \cup \{nil\}$ is a parameter if applicable,
\item for any two gates $T =(G,v_{in_1}, v_{in_2},v_{out},\alpha)$ and $T'=(G',v_{in_1}', v_{in_2}',v_{out}',\alpha')$ in $\TT$ where $T \neq T'$, they must satisfy $v_{out} \neq v_{out}'$. 
\end{itemize} 
Note that generalized circuits extend the standard boolean or arithmetic circuits in the sense that generalized circuits allow cycles in the directed graph defined by the nodes and gates. We define the search problem $\epsilon$-\GC\ \cite{chen2006settling,rubinstein2015inapproximability}:

\begin{problem}{$\epsilon$-\GC}
	
	\textbf{Input:} A generalized circuit $S=(V,\TT)$. 
	
	\textbf{Output:} A vector $\mathbf{x} \in [0,1]^N$ of values for the nodes, satisfying the conditions from Table \ref{table:gcircuit}.
\end{problem}	
Note that a solution to $\epsilon$-\GC\ always exists \cite{chen2006settling} and hence the problem is well-defined. Also, notice that for the logic gates $G_{\vee},G_{\wedge}$ and $G_{\neg}$, if the input conditions are not fulfilled, the output is unconstrained, and for the multiplication gate, it holds that $\alpha \in (0,1]$. $\epsilon$-\GC\ was proven to be PPAD-complete implicitly or explicitly in \cite{daskalakis2009complexity,chen2006settling} for inversely polynomial error $\epsilon$ and in \cite{rubinstein2015inapproximability} for constant $\epsilon$. We state the latter theorem here as a lemma:

\begin{small}
\begin{table} [!htbp]
\centering 
\caption{Gate constraint $T=(G,v_{in_1}, v_{in_2},v_{out},\alpha)$}
\label{table:gcircuit}
\begin{tabular}{| c | c |} 
\hline
Gate & Constraint \\
\hline
$(G_{\zeta},nil,nil,v_{out},\alpha)$ & $\alpha - \epsilon \leq \xx[v_{out}] \leq \alpha + \epsilon$ \\
\hline
$(G_{\times \zeta},v_{in_1},nil,v_{out},\alpha)$ & $\alpha \cdot \xx[v_{in_1}]- \epsilon \leq \xx[v_{out}] \leq \alpha \cdot \xx[v_{in_1}] + \epsilon$ \\
\hline 
$(G_{=},v_{in_1},nil,v_{out},nil)$ & $\xx [v_{in_1}]- \epsilon \leq \xx[v_{out}] \leq \xx[v_{in_1}] + \epsilon$ \\
\hline
$(G_{+},v_{in_1}, v_{in_2},v_{out},nil)$ & $ \xx[v_{out}] \in \left[\min(\xx[v_{in_1}]+\xx[v_{in_2}], 1) -  \epsilon, 
																	\min(\xx[v_{in_1}]+\xx[v_{in_2}], 1) + \epsilon\right]$ \\
\hline
$(G_{-},v_{in_1}, v_{in_2},v_{out},nil)$ & $ \xx[v_{out}] \in \left[\max(\xx[v_{in_1}]-\xx[v_{in_2}], 0) -  \epsilon, 
																\max(\xx[v_{in_1}]-\xx[v_{in_2}], 0) + \epsilon\right]$ \\
\hline
$(G_{<},v_{in_1}, v_{in_2},v_{out},nil)$ & $\xx[v_{out}] = \left\{
	\begin{array}{ll}
		1 \pm \epsilon,  	\,\,\,\,\,\ & \hbox{if $\xx[v_{in_1}]<\xx[v_{in_2}] - \epsilon$;} \\
		0 \pm \epsilon, 	\,\,\,\,\,\ & \hbox{if $\xx[v_{in_1}]>\xx[v_{in_2}] + \epsilon$.}
	\end{array} \right.$ \\
\hline
$(G_{\vee},v_{in_1}, v_{in_2},v_{out},nil)$ & $\xx[v_{out}] = \left\{
	\begin{array}{ll}
		1 \pm \epsilon,  	\,\,\,\,\,\ & \hbox{if $\xx[v_{in_1}]=1 \pm \epsilon \,\,\ \text{or} 	\,\,\ \xx[v_{in_2}]=1 \pm \epsilon$;} \\
		0 \pm \epsilon, 	\,\,\,\,\,\ & \hbox{if $\xx[v_{in_1}]=0 \pm \epsilon \,\,\ \text{and} \,\,\ \xx[v_{in_2}]=0 \pm \epsilon$.}
	\end{array} \right.$ \\
\hline
$(G_{\wedge},v_{in_1}, v_{in_2},v_{out},nil)$ & $\xx[v_{out}] = \left\{
	\begin{array}{ll}
		1 \pm \epsilon,  	\,\,\,\,\,\ & \hbox{if $\xx[v_{in_1}]=1 \pm \epsilon \,\,\ \text{and} \,\,\ \xx[v_{in_2}]=1 \pm \epsilon$;} \\
		0 \pm \epsilon, 	\,\,\,\,\,\ & \hbox{if $\xx[v_{in_1}]=0 \pm \epsilon \,\,\ \text{or} 	\,\,\ \xx[v_{in_2}]=0 \pm \epsilon$.}
	\end{array} \right.$ \\
\hline
$(G_{\neg},v_{in_1},nil,v_{out},nil)$ & $\xx[v_{out}] = \left\{
	\begin{array}{ll}
		1 \pm \epsilon,  	\,\,\,\,\,\ & \hbox{if $\xx[v_{in_1}]=0 \pm \epsilon$;} \\
		0 \pm \epsilon, 	\,\,\,\,\,\ & \hbox{if $\xx[v_{in_1}]=1 \pm \epsilon$.}
	\end{array} \right.$ \\
\hline
\end{tabular}
\end{table}
\end{small}

\begin{lemma}[\cite{rubinstein2015inapproximability}]\label{lem:gcppad}
	There exists a constant $\epsilon >0$ such that $\epsilon$-\GC\ is PPAD-complete. 
\end{lemma}

\vspace{0.1cm}

\section{$(n,n+k,\epsilon)$-\CSH\ is PPAD-hard}\label{sec:PPADhard}


In this section, we will prove that finding an approximate partition for Consensus-halving using $n$ cuts is PPAD-hard, even if the allowed discrepancy between the two portions is a constant. We describe the reduction from $\epsilon$-\GC\ that we will be using for the PPAD-hardness proof. Given an instance $S=(V, \TT)$ of $\epsilon$-\GC, we will construct an instance of $(n,n,\epsilon')$-\CSH\ with $n=2N$ agents, in which each node $v_i \in V$ of the circuit will correspond to two agents $var_i$ and $copy_i$ and where $\epsilon'$ will be defined later. As a matter of convenience in the reduction, we will assume that for every gate $(G,v_{in_1}, v_{in_2},v_{out},\alpha)$ in $\TT$, $v_{in_1}, v_{in_2}$ and $v_{out}$ are distinct. This does not affect the hardness of the problem as any $\epsilon$-generalized circuit can be converted to this form by use of at most $2N$ additional equality-gates and nodes, and also since an $(\epsilon/2)$-approximate solution to the converted problem can clearly be converted to a solution in the original circuit. 

For ease of notation we consider a \CSH\ instance on the interval $[0,6N]$. Let $d_i := 6(i-1)$; the two agents $var_i$ and $copy_i$ that we construct for every node $v_i$ have valuations %
\begin{ceqn}
\begin{align*}
var_i  &= \begin{cases}
border_i(t) + G^{\tau}(t)				,& \text{if } v_i \text{ is the output of }\tau\\
border_i(t) 	,& \text{otherwise} \end{cases}\\
copy_i  &= \begin{cases}
4				,& t \in [d_i+3,d_i +4] \cup [d_i+5,d_i +6]\\
1				,& t \in [d_i+1,d_i +2] \cup [d_i+4,d_i +5]\\
0				,& \text{otherwise} \end{cases}\\
\textrm{where }\ \ \ border_i  &= \begin{cases}
4	,& \text{if } t \in [d_i,d_i +1] \cup [d_i+2,d_i +3]\\
0 	,& \text{otherwise} \end{cases}
\end{align*}
\end{ceqn}
Since each node is the output of at most one gate, $var_i$ is well-defined. Note that apart from the valuation defined by the function $G^{\tau}$, agents $var_i$ and $copy_i$ only have valuations on the sub-interval $[d_i,d_{i+1}]$, i.e., the agents associated with node $v_1$ only have valuations on $[0,6]$, the agents associated with $v_2$ only on have valuations on $[6,12]$ and so on. Let $v_i^{-}:=[d_i +1,d_i+2]$ and the right and left endpoints respectively be $v_{i,\ell}^{-}$ and $v_{i,r}^{-}$, (and analogously for $v_i^{+}:=[d_i +3,d_i+4]$, $v_{i,\ell}^{+}$ and $v_{i,r}^{+}$). Now, we are ready to define the functions $G^\tau$ according to Table \ref{table:gate-agents}. Notice that because of the assumption that the two input nodes and the output node are distinct, the graphs of the functions are as in Table \ref{table:gate-agents}. Figure \ref{fig:circuit} demonstrates an example of a Consensus-halving instance corresponding to a small circuit. 

\begin{table}[!ht]
	\caption{Agent preferences from gate $\tau=(G,v_{in_1}, v_{in_2},v_{out},\alpha)$. For the gate $G_{\times \zeta}$, the figure depicts the case when $\alpha+\epsilon<1$.}
	\label{table:gate-agents} 
	\centering 
	\begin{tabular}{| >{\centering\arraybackslash}m{0.25in} | >{\arraybackslash}m{2.4in} | >{\centering\arraybackslash}m{1.87in} |} 
		\hline
		& $ \hspace{1.1in} G^{\tau}(t) $ & Picture \\
		\hline
		$\begin{scriptsize} G_{\zeta} \end{scriptsize}$ 
		& $\begin{scriptsize} \begin{aligned}
		\begin{cases}
		1								& \text{if } t \in [v_{out,\ell}^{-} + \alpha - \frac{1}{2} , v_{out,\ell}^{-} + \alpha + \frac{1}{2}] \\
		0								& \text{otherwise}
		\end{cases}
		\end{aligned} \end{scriptsize}$
		&  \begin{tikzpicture}
	\node (up) at (0pt,20pt) {};
	\node (a_1) at (0pt,0pt) {}; 
	\node (a_2) at (130pt, 0pt) {};
	\draw (a_1)--(a_2);
	\draw[fill=lightgray] (50pt,0pt) rectangle (71pt,10pt);
	
	\draw [
    thick,
    decoration={
        brace,
        mirror,
        raise=5pt
    },
    decorate
] (57pt,0pt) -- (78pt,0pt)
node [pos=0.5,anchor=north,yshift=-5pt] {$v_{out}^{-}$};
\draw[densely dotted] (71pt,14pt) -- (71pt,-10pt);
\draw[densely dotted] (57pt,14pt) -- (57pt,-10pt);

\draw[densely dotted] (78pt,14pt) -- (78pt,-10pt);

\draw [overbrace style] (57pt,7pt) -- (71pt,7pt) node [overbrace text style] {$\alpha+\frac{1}{2}$};
\end{tikzpicture} \\
		\hline
		$\begin{scriptsize} G_{\times \zeta} \end{scriptsize}$
		& $\begin{scriptsize} \begin{aligned}
		\begin{cases}
		1								& \text{if } t \in v_{in}^+ \\
		1/\alpha				& \text{if } t \in [v_{out,\ell}^{-}, v_{out,\ell}^{-} + \min(\alpha + \epsilon,1)] \\
		0								& \text{otherwise} 
		\end{cases}				
		\end{aligned} \end{scriptsize}$
		& \begin{tikzpicture}
	\node (up) at (0pt,20pt) {};
	\node (a_1) at (0pt,0pt) {}; 
	\node (a_2) at (130pt, 0pt) {};
	\draw (a_1)--(a_2);
	\draw[fill=lightgray] (36pt,0pt) rectangle (57pt,10pt);
	\draw[fill=lightgray] (78pt,0pt) rectangle (92pt,15pt);

	\draw[<->,densely dotted] (78pt,18pt) -- (92pt,18pt);
	\node[above] (alpha) at (85pt,16pt) {\tiny{$\alpha+\epsilon$}};
	
	\draw[<->,densely dotted] (96pt,15pt) -- (96pt,0pt);
	\node[above] (alpha_2) at (106pt,0pt) {\tiny{$1/\alpha$}};

	\draw [
    thick,
    decoration={
        brace,
        mirror,
        raise=5pt
    },
    decorate
] (36pt,0pt) -- (57pt,0pt)
node [pos=0.5,anchor=north,yshift=-5pt] {$v_{in}^{+}$}; 

\draw [
    thick,
    decoration={
        brace,
        mirror,
        raise=5pt
    },
    decorate
] (78pt,0pt) -- (99pt,0pt)
node [pos=0.5,anchor=north,yshift=-5pt] {$v_{out}^{-}$};
	\end{tikzpicture} \\
		\hline	
		$\begin{scriptsize} G_{\neg} \end{scriptsize}$ 
		& $\begin{scriptsize} \begin{aligned}
		\begin{cases}
		1								& \text{if } t \in v_{in}^{-} \\
		1/2\epsilon			& \text{if } t \in [v_{out,\ell}^{-},v_{out,\ell}^{-}+\epsilon] \\
		1/2\epsilon			& \text{if } t \in [v_{out,r}^{-}-\epsilon,v_{out,r}^{-}] \\
		0								& \text{otherwise}   
		\end{cases}
		\end{aligned} \end{scriptsize}$
		& \begin{tikzpicture}
	\node (up) at (0pt,20pt) {};
	\node (a_1) at (0pt,0pt) {}; 
	\node (a_2) at (130pt, 0pt) {};
	\draw (a_1)--(a_2);
	\draw[fill=lightgray] (36pt,0pt) rectangle (57pt,10pt);
	
	\fill[fill=lightgray] (78pt,0pt) rectangle (82.5pt,24.2pt);
	\fill[fill=lightgray] (82.5pt,0pt) rectangle (110pt,0pt);
	\draw[fill=lightgray] (78pt,24.2pt) -- (78pt,0pt);
	\draw[fill=lightgray] (78pt,24.2pt) -- (82.5pt,24.2pt);
	\draw[fill=lightgray] (82.5pt,24.2pt) -- (82.5pt,0pt);
	\draw[fill=lightgray] (82.5pt,0pt) -- (99pt,0pt);
	\fill[fill=lightgray] (78pt,0pt) rectangle (94.5pt,0pt);
	\fill[fill=lightgray] (94.5pt,0pt) rectangle (99pt,24.2pt);
	\draw (78pt,0pt) -- (78pt,0pt);
	\draw (82.5pt,0pt) -- (94.5pt,0pt);
	\draw (94.5pt,0pt) -- (94.5pt,24.2pt);
	\draw (94.5pt,24.2pt) -- (99pt,24.2pt);
	\draw (99pt,24.2pt) -- (99pt,0pt);
	\draw (78pt,0pt) -- (99pt,0pt);
	
	\draw [
    thick,
    decoration={
        brace,
        mirror,
        raise=5pt
    },
    decorate
] (36pt,0pt) -- (57pt,0pt)
node [pos=0.5,anchor=north,yshift=-5pt] {$v_{in1}^{-}$}; 

\draw [
    thick,
    decoration={
        brace,
        mirror,
        raise=5pt
    },
    decorate
] (78pt,0pt) -- (99pt,0pt)
node [pos=0.5,anchor=north,yshift=-5pt] {$v_{out}^{-}$};

	\draw[<->] (78pt,27pt) -- (82.5pt,27pt);
	\node[above] (alpha) at (79.5pt,28pt) {\tiny{$\epsilon$}};
	
	\draw[<->] (94.5pt,27pt) -- (99pt,27pt);
	\node[above] (alpha) at (96pt,28pt) {\tiny{$\epsilon$}};
	
		\draw[<->] (102pt,24.2pt) -- (102pt,0pt);
		\node[above] (alpha_5) at (107pt,5pt) {\tiny{$\frac{1}{2\epsilon}$}};
	\end{tikzpicture} \\
		\hline 
		$\begin{scriptsize} G_{+} \end{scriptsize}$ 
		& $\begin{scriptsize} \begin{aligned}
		\begin{cases}
		1								& \text{if } t \in v_{in_1}^+ \cup v_{in_2}^+ \\
		1								& \text{if } t \in [v_{out,\ell}^{-},v_{out,r}^{-} - \epsilon] \\
		1/\epsilon +1		& \text{if } t \in [v_{out,r}^{-} - \epsilon,v_{out,r}^{-}] \\
		0								& \text{otherwise} 
		\end{cases}		
		\end{aligned} \end{scriptsize}$
		& \begin{tikzpicture}
	\node (up) at (0pt,20pt) {};
	\node (a_1) at (0pt,0pt) {}; 
	\node (a_2) at (130pt, 0pt) {};
	\draw (a_1)--(a_2);
	\draw[fill=lightgray] (15pt,0pt) rectangle (36pt,10pt);
	\draw[fill=lightgray] (52pt,0pt) rectangle (73pt,10pt);
	\fill[fill=lightgray] (89pt,0pt) rectangle (105.5pt,10pt);
	\fill[fill=lightgray] (105.5pt,0pt) rectangle (110pt,22.2pt);
	\draw (89pt,10pt) -- (89pt,0pt);
	\draw (89pt,10pt) -- (105.5pt,10pt);
	\draw (105.5pt,10pt) -- (105.5pt,22.2pt);
	\draw (105.5pt,22.2pt) -- (110pt,22.2pt);
	\draw (110pt,22.2pt) -- (110pt,0pt);
	\draw (89pt,0pt) -- (110pt,0pt);
	
	\draw [
    thick,
    decoration={
        brace,
        mirror,
        raise=5pt
    },
    decorate
] (15pt,0pt) -- (36pt,0pt)
node [pos=0.5,anchor=north,yshift=-5pt] {$v_{in1}^{+}$}; 

\draw [
    thick,
    decoration={
        brace,
        mirror,
        raise=5pt
    },
    decorate
] (52pt,0pt) -- (73pt,0pt)
node [pos=0.5,anchor=north,yshift=-5pt] {$v_{in2}^{+}$};

\draw [
    thick,
    decoration={
        brace,
        mirror,
        raise=5pt
    },
    decorate
] (89pt,0pt) -- (110pt,0pt)
node [pos=0.5,anchor=north,yshift=-5pt] {$v_{out}^{-}$};
	
	\draw[<->] (105.5pt,25pt) -- (110pt,25pt);
	\node[above] (alpha) at (107pt,25pt) {\tiny{$\epsilon$}};
	
	\draw[<->] (113pt,22.2pt) -- (113pt,10pt);
	\node[above] (alpha_2) at (119pt,8pt) {\tiny{$\frac{1}{\epsilon}$}}; 
	
	\draw[densely dotted] (113pt,22.2pt) -- (125.5pt,22.2pt);
	\draw[densely dotted] (113pt,10pt) -- (125.5pt,10pt);
	\end{tikzpicture} \\
		\hline
		$\begin{scriptsize} G_{-} \end{scriptsize}$ 
		& $\begin{scriptsize} \begin{aligned}
		\begin{cases}
		1								& \text{if } t \in v_{in_1}^+ \cup v_{in_2}^{-} \\
		1								& \text{if } t \in [v_{out,\ell}^{-}+\epsilon,v_{out,r}^{-}] \\
		1/\epsilon +1		& \text{if } t \in [v_{out,\ell}^{-},v_{out,\ell}^{-}+\epsilon] \\
		0								& \text{otherwise} 
		\end{cases}
		\end{aligned} \end{scriptsize}$
		& \begin{tikzpicture}
		\node (up) at (0pt,20pt) {};
	\node (a_1) at (0pt,0pt) {}; 
	\node (a_2) at (130pt, 0pt) {};
	\draw (a_1)--(a_2);
	\draw[fill=lightgray] (15pt,0pt) rectangle (36pt,10pt);
	\draw[fill=lightgray] (52pt,0pt) rectangle (73pt,10pt);
	\fill[fill=lightgray] (89pt,0pt) rectangle (93.5pt,22.2pt);
	\fill[fill=lightgray] (93.5pt,0pt) rectangle (110pt,10pt);
	\draw (89pt,22.2pt) -- (89pt,0pt);
	\draw (89pt,22.2pt) -- (93.5pt,22.2pt);
	\draw (93.5pt,22.2pt) -- (93.5pt,10pt);
	\draw (93.5pt,10pt) -- (110pt,10pt);
	\draw (110pt,10pt) -- (110pt,0pt);
	\draw (89pt,0pt) -- (110pt,0pt);

	\draw [
    thick,
    decoration={
        brace,
        mirror,
        raise=5pt
    },
    decorate
] (15pt,0pt) -- (36pt,0pt)
node [pos=0.5,anchor=north,yshift=-5pt] {$v_{in1}^{+}$}; 

\draw [
    thick,
    decoration={
        brace,
        mirror,
        raise=5pt
    },
    decorate
] (52pt,0pt) -- (73pt,0pt)
node [pos=0.5,anchor=north,yshift=-5pt] {$v_{in2}^{-}$};

\draw [
    thick,
    decoration={
        brace,
        mirror,
        raise=5pt
    },
    decorate
] (89pt,0pt) -- (110pt,0pt)
node [pos=0.5,anchor=north,yshift=-5pt] {$v_{out}^{-}$};
	
	\draw[<->] (89pt,25pt) -- (93.5pt,25pt);
	\node[above] (alpha) at (90.5pt,25pt) {\tiny{$\epsilon$}};
	
	\draw[<->] (86pt,22.2pt) -- (86pt,10pt);
	
		\node[above] (alpha_2) at (81pt,8pt) {\tiny{$\frac{1}{\epsilon}$}}; 
		
		\draw[densely dotted] (75pt,22.2pt) -- (86.5pt,22.2pt);
		\draw[densely dotted] (75pt,10pt) -- (86.5pt,10pt);
	\end{tikzpicture} \\
		\hline
		$\begin{scriptsize} G_{<} \end{scriptsize}$ 
		& $\begin{scriptsize} \begin{aligned}
		\begin{cases}
		1								& \text{if } t \in v_{in_1}^+ \cup v_{in_2}^{-} \\
		1/\epsilon			& \text{if } t \in [v_{out,\ell}^{-},v_{out,\ell}^{-}+\epsilon] \\
		1/\epsilon			& \text{if } t \in [v_{out,r}^{-}-\epsilon,v_{out,r}^{-}] \\
		0								& \text{otherwise} 
		\end{cases}
		\end{aligned} \end{scriptsize}$
		& \begin{tikzpicture}
		\node (up) at (0pt,20pt) {};
	\node (a_1) at (0pt,0pt) {}; 
	\node (a_2) at (130pt, 0pt) {};
	\draw (a_1)--(a_2);
	\draw[fill=lightgray] (15pt,0pt) rectangle (36pt,10pt);
	\draw[fill=lightgray] (52pt,0pt) rectangle (73pt,10pt);
	\fill[fill=lightgray] (89pt,0pt) rectangle (93.5pt,24.2pt);
	\fill[fill=lightgray] (93.5pt,0pt) rectangle (110pt,0pt);
	\draw[fill=lightgray] (89pt,24.2pt) -- (89pt,0pt);
	\draw[fill=lightgray] (89pt,24.2pt) -- (93.5pt,24.2pt);
	\draw[fill=lightgray] (93.5pt,24.2pt) -- (93.5pt,0pt);
	\draw[fill=lightgray] (93.5pt,0pt) -- (110pt,0pt);
	\draw[fill=lightgray] (15pt,0pt) rectangle (36pt,10pt);
	\draw[fill=lightgray] (52pt,0pt) rectangle (73pt,10pt);
	\fill[fill=lightgray] (89pt,0pt) rectangle (105.5pt,0pt);
	\fill[fill=lightgray] (105.5pt,0pt) rectangle (110pt,24.2pt);
	\draw (89pt,0pt) -- (89pt,0pt);
	\draw (93.5pt,0pt) -- (105.5pt,0pt);
	\draw (105.5pt,0pt) -- (105.5pt,24.2pt);
	\draw (105.5pt,24.2pt) -- (110pt,24.2pt);
	\draw (110pt,24.2pt) -- (110pt,0pt);
	\draw (89pt,0pt) -- (110pt,0pt);

	\draw [
    thick,
    decoration={
        brace,
        mirror,
        raise=5pt
    },
    decorate
] (15pt,0pt) -- (36pt,0pt)
node [pos=0.5,anchor=north,yshift=-5pt] {$v_{in1}^{+}$}; 

\draw [
    thick,
    decoration={
        brace,
        mirror,
        raise=5pt
    },
    decorate
] (52pt,0pt) -- (73pt,0pt)
node [pos=0.5,anchor=north,yshift=-5pt] {$v_{in2}^{-}$};

\draw [
    thick,
    decoration={
        brace,
        mirror,
        raise=5pt
    },
    decorate
] (89pt,0pt) -- (110pt,0pt)
node [pos=0.5,anchor=north,yshift=-5pt] {$v_{out}^{-}$};
	
	\draw[<->] (89pt,27pt) -- (93.5pt,27pt);
	\node[above] (alpha) at (90.5pt,28pt) {\tiny{$\epsilon$}};
	
	\draw[<->] (105.5pt,27pt) -- (110pt,27pt);
	\node[above] (alpha_3) at (107pt,27pt) {\tiny{$\epsilon$}};
	
	
	\draw[<->] (113pt,24.2pt) -- (113pt,0pt);
	\node[above] (alpha_5) at (118pt,5pt) {\tiny{$\frac{1}{\epsilon}$}};
	
	\end{tikzpicture} \\
		\hline
		$\begin{scriptsize} G_{\vee} \end{scriptsize}$ 
		& $\begin{scriptsize} \begin{aligned}
		\begin{cases}
		1								& \text{if } t \in v_{in_1}^+ \cup v_{in_2}^{+} \\
		0.5/\epsilon		& \text{if } t \in [v_{out,\ell}^{-},v_{out,\ell}^{-}+\epsilon] \\
		1.5/\epsilon		& \text{if } t \in [v_{out,r}^{-}-\epsilon,v_{out,r}^{-}] \\
		0								& \text{otherwise} 
		\end{cases}
		\end{aligned} \end{scriptsize}$
		& \begin{tikzpicture}
		\node (up) at (0pt,20pt) {};
	\node (a_1) at (0pt,0pt) {}; 
	\node (a_2) at (130pt, 0pt) {};
	\draw (a_1)--(a_2);
	\draw[fill=lightgray] (15pt,0pt) rectangle (36pt,10pt);
	\draw[fill=lightgray] (52pt,0pt) rectangle (73pt,10pt);
	\fill[fill=lightgray] (89pt,0pt) rectangle (93.5pt,14.2pt);
	\fill[fill=lightgray] (93.5pt,0pt) rectangle (110pt,0pt);
	\draw[fill=lightgray] (89pt,14.2pt) -- (89pt,0pt);
	\draw[fill=lightgray] (89pt,14.2pt) -- (93.5pt,14.2pt);
	\draw[fill=lightgray] (93.5pt,14.2pt) -- (93.5pt,0pt);
	\draw[fill=lightgray] (93.5pt,0pt) -- (110pt,0pt);
	\draw[fill=lightgray] (15pt,0pt) rectangle (36pt,10pt);
	\draw[fill=lightgray] (52pt,0pt) rectangle (73pt,10pt);
	\fill[fill=lightgray] (89pt,0pt) rectangle (105.5pt,0pt);
	\fill[fill=lightgray] (105.5pt,0pt) rectangle (110pt,24.2pt);
	\draw (89pt,0pt) -- (89pt,0pt);
	\draw (93.5pt,0pt) -- (105.5pt,0pt);
	\draw (105.5pt,0pt) -- (105.5pt,24.2pt);
	\draw (105.5pt,24.2pt) -- (110pt,24.2pt);
	\draw (110pt,24.2pt) -- (110pt,0pt);
	\draw (89pt,0pt) -- (110pt,0pt);

	\draw [
    thick,
    decoration={
        brace,
        mirror,
        raise=5pt
    },
    decorate
] (15pt,0pt) -- (36pt,0pt)
node [pos=0.5,anchor=north,yshift=-5pt] {$v_{in1}^{+}$}; 

\draw [
    thick,
    decoration={
        brace,
        mirror,
        raise=5pt
    },
    decorate
] (52pt,0pt) -- (73pt,0pt)
node [pos=0.5,anchor=north,yshift=-5pt] {$v_{in2}^{+}$};

\draw [
    thick,
    decoration={
        brace,
        mirror,
        raise=5pt
    },
    decorate
] (89pt,0pt) -- (110pt,0pt)
node [pos=0.5,anchor=north,yshift=-5pt] {$v_{out}^{-}$};
	
	\draw[<->] (89pt,17pt) -- (93.5pt,17pt);
	\node[above] (alpha) at (90.5pt,18pt) {\tiny{$\epsilon$}};
	
	\draw[<->] (86pt,14.2pt) -- (86pt,0pt);
	\node[above] (alpha_2) at (80pt,0pt) {\tiny{$\frac{0.5}{\epsilon}$}}; 
	
	\draw[<->] (105.5pt,27pt) -- (110pt,27pt);
	\node[above] (alpha_3) at (107pt,27pt) {\tiny{$\epsilon$}};
	
	
	\draw[<->] (113pt,24.2pt) -- (113pt,0pt);
	\node[above] (alpha_5) at (120pt,5pt) {\tiny{$\frac{1.5}{\epsilon}$}};
	
	
	
	\end{tikzpicture} \\
		\hline
		$\begin{scriptsize} G_{\wedge} \end{scriptsize}$ 
		&	$\begin{scriptsize}
		\begin{cases}
		1								& \text{if } t \in v_{in_1}^+ \cup v_{in_2}^{+} \\
		1.5/\epsilon		& \text{if } t \in [v_{out,\ell}^{-},v_{out,\ell}^{-}+\epsilon] \\
		0.5/\epsilon		& \text{if } t \in [v_{out,r}^{-}-\epsilon,v_{out,r}^{-}] \\
		0								& \text{otherwise} 
		\end{cases}
		\end{scriptsize}$
		& \begin{tikzpicture}
		\node (up) at (0pt,20pt) {};
	\node (a_1) at (0pt,0pt) {}; 
	\node (a_2) at (130pt, 0pt) {};
	\draw (a_1)--(a_2);
	\draw[fill=lightgray] (15pt,0pt) rectangle (36pt,10pt);
	\draw[fill=lightgray] (52pt,0pt) rectangle (73pt,10pt);
	\fill[fill=lightgray] (89pt,0pt) rectangle (93.5pt,24.2pt);
	\fill[fill=lightgray] (93.5pt,0pt) rectangle (110pt,0pt);
	\draw[fill=lightgray] (89pt,24.2pt) -- (89pt,0pt);
	\draw[fill=lightgray] (89pt,24.2pt) -- (93.5pt,24.2pt);
	\draw[fill=lightgray] (93.5pt,24.2pt) -- (93.5pt,0pt);
	\draw[fill=lightgray] (93.5pt,0pt) -- (110pt,0pt);
	\draw[fill=lightgray] (15pt,0pt) rectangle (36pt,10pt);
	\draw[fill=lightgray] (52pt,0pt) rectangle (73pt,10pt);
	\fill[fill=lightgray] (89pt,0pt) rectangle (105.5pt,0pt);
	\fill[fill=lightgray] (105.5pt,0pt) rectangle (110pt,14.2pt);
	\draw (89pt,0pt) -- (89pt,0pt);
	\draw (93.5pt,0pt) -- (105.5pt,0pt);
	\draw (105.5pt,0pt) -- (105.5pt,14.2pt);
	\draw (105.5pt,14.2pt) -- (110pt,14.2pt);
	\draw (110pt,14.2pt) -- (110pt,0pt);
	\draw (89pt,0pt) -- (110pt,0pt);

	\draw [
    thick,
    decoration={
        brace,
        mirror,
        raise=5pt
    },
    decorate
] (15pt,0pt) -- (36pt,0pt)
node [pos=0.5,anchor=north,yshift=-5pt] {$v_{in1}^{+}$}; 

\draw [
    thick,
    decoration={
        brace,
        mirror,
        raise=5pt
    },
    decorate
] (52pt,0pt) -- (73pt,0pt)
node [pos=0.5,anchor=north,yshift=-5pt] {$v_{in2}^{+}$};

\draw [
    thick,
    decoration={
        brace,
        mirror,
        raise=5pt
    },
    decorate
] (89pt,0pt) -- (110pt,0pt)
node [pos=0.5,anchor=north,yshift=-5pt] {$v_{out}^{-}$};
	
	\draw[<->] (89pt,27pt) -- (93.5pt,27pt);
	\node[above] (alpha) at (90.5pt,28pt) {\tiny{$\epsilon$}};
	
	\draw[<->] (86pt,24.2pt) -- (86pt,0pt);
	\node[above] (alpha_2) at (80pt,5pt) {\tiny{$\frac{1.5}{\epsilon}$}}; 
	
	\draw[<->] (105.5pt,17pt) -- (110pt,17pt);
	\node[above] (alpha_3) at (107pt,17pt) {\tiny{$\epsilon$}};
	
	
	\draw[<->] (113pt,14.2pt) -- (113pt,0pt);
	\node[above] (alpha_5) at (120pt,0pt) {\tiny{$\frac{0.5}{\epsilon}$}};
	
	
	
	\end{tikzpicture} \\
		\hline
	\end{tabular}
\end{table}

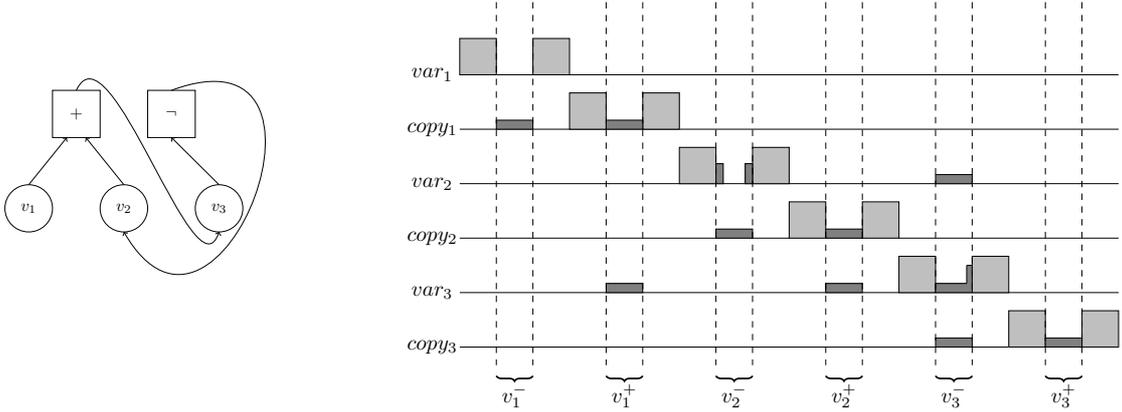
\begin{figure}
	\centering
	\begin{tabular}{cc}
		\begin{minipage}{0.3\textwidth}
			\resizebox{1\linewidth}{!}{\begin{tikzpicture}
\draw (0,0) circle (0.5cm) node { $v_1$};
\draw (2,0) circle (0.5cm) node { $v_2$};
\draw (4,0) circle (0.5cm) node { $v_3$};

\draw (0.5,1.5) rectangle (1.5,2.5) node[pos=.5] {$+$};
\draw (2.5,1.5) rectangle (3.5,2.5) node[pos=.5] {$\lnot$};

\draw[->] (0,0.5) -- (0.8,1.5);
\draw[->] (2,0.5) -- (1.2,1.5);

\draw[->] (4,0.5) -- (3,1.5);

\draw [->] (1,2.5) .. controls (1.7,4) and (3.5,-2) .. (4,-0.5);
\draw [->] (3,2.5) .. controls (7,4) and (4,-4) .. (2,-0.5);
\end{tikzpicture}}
		\end{minipage} &
		\begin{minipage}{0.6\textwidth}
			\resizebox{1\linewidth}{!}{\begin{tikzpicture}[scale=0.3]
\node (a_1) at (-1.5,15) {$var_1$}; 
\node (a_2) at (-1.5,12) {$copy_1$}; 

\node (b_1) at (-1.5,9) {$var_2$}; 
\node (b_2) at (-1.5,6) {$copy_2$}; 

\node (c_1) at (-1.5,3) {$var_3$}; 
\node (c_2) at (-1.5,0) {$copy_3$}; 

\draw (0,0) -- (36,0);
\draw (0,3) -- (36,3);
\draw (0,6) -- (36,6);
\draw (0,9) -- (36,9);
\draw (0,12) -- (36,12);
\draw (0,15) -- (36,15);

\draw[fill=lightgray] (0,15) rectangle (2,17);
\draw[fill=lightgray] (4,15) rectangle (6,17);

\draw[fill=lightgray] (6,12) rectangle (8,14);
\draw[fill=lightgray] (10,12) rectangle (12,14);

\draw[fill=lightgray] (12,9) rectangle (14,11);
\draw[fill=lightgray] (16,9) rectangle (18,11);

\draw[fill=lightgray] (18,6) rectangle (20,8);
\draw[fill=lightgray] (22,6) rectangle (24,8);

\draw[fill=lightgray] (24,3) rectangle (26,5);
\draw[fill=lightgray] (28,3) rectangle (30,5);

\draw[fill=lightgray] (30,0) rectangle (32,2);
\draw[fill=lightgray] (34,0) rectangle (36,2);

\fill[fill=gray] (26,3) rectangle (27.7,3.5);
\fill[fill=gray] (27.7,3) rectangle (28,4.5);
\draw (26,3.5) -- (26,3);
\draw (26,3.5) -- (27.7,3.5);
\draw (27.7,3.5) -- (27.7,4.5);
\draw (27.7,4.5) -- (28,4.5);
\draw (28,4.5) -- (28,3);
\draw (26,3) -- (28,3);

\draw[fill=gray] (8,3) rectangle (10,3.5);
\draw[fill=gray] (20,3) rectangle (22,3.5);

\draw[fill=gray] (14,9) rectangle (14.4,10.1);
\draw[fill=gray] (15.6,9) rectangle (16,10.1);

\draw[fill=gray] (26,9) rectangle (28,9.5);

\draw[fill=gray] (2,12) rectangle (4,12.5);
\draw[fill=gray] (8,12) rectangle (10,12.5);

\draw[fill=gray] (14,6) rectangle (16,6.5);
\draw[fill=gray] (20,6) rectangle (22,6.5);

\draw[fill=gray] (26,0) rectangle (28,0.5);
\draw[fill=gray] (32,0) rectangle (34,0.5);

\draw[dashed] (2,19) -- (2,-1);
\draw[dashed] (4,19) -- (4,-1);

\draw[dashed] (8,19) -- (8,-1);
\draw[dashed] (10,19) -- (10,-1);

\draw[dashed] (14,19) -- (14,-1);
\draw[dashed] (16,19) -- (16,-1);

\draw[dashed] (20,19) -- (20,-1);
\draw[dashed] (22,19) -- (22,-1);

\draw[dashed] (26,19) -- (26,-1);
\draw[dashed] (28,19) -- (28,-1);

\draw[dashed] (32,19) -- (32,-1);
\draw[dashed] (34,19) -- (34,-1);

\draw [
thick,
decoration={
	brace,
	mirror,
	raise=5pt
},
decorate
] (2,-1) -- (4,-1)
node [pos=0.5,anchor=north,yshift=-5pt] {$v_{1}^{-}$};

\draw [
thick,
decoration={
	brace,
	mirror,
	raise=5pt
},
decorate
] (8,-1) -- (10,-1)
node [pos=0.5,anchor=north,yshift=-5pt] {$v_{1}^{+}$};

\draw [
thick,
decoration={
	brace,
	mirror,
	raise=5pt
},
decorate
] (14,-1) -- (16,-1)
node [pos=0.5,anchor=north,yshift=-5pt] {$v_{2}^{-}$};

\draw [
thick,
decoration={
	brace,
	mirror,
	raise=5pt
},
decorate
] (20,-1) -- (22,-1)
node [pos=0.5,anchor=north,yshift=-5pt] {$v_{2}^{+}$};

\draw [
thick,
decoration={
	brace,
	mirror,
	raise=5pt
},
decorate
] (26,-1) -- (28,-1)
node [pos=0.5,anchor=north,yshift=-5pt] {$v_{3}^{-}$};

\draw [
thick,
decoration={
	brace,
	mirror,
	raise=5pt
},
decorate
] (32,-1) -- (34,-1)
node [pos=0.5,anchor=north,yshift=-5pt] {$v_{3}^{+}$};

\end{tikzpicture}}
		\end{minipage}
	\end{tabular}
	\caption{An instance of $\epsilon$-{\sc{Gcircuit}} with the corresponding construction for $(n,\epsilon')$-{\sc{Chalving}}.}
	\label{fig:circuit}
\end{figure}

\begin{lemma}\label{lem:reduc}
	Given the construction of a $(n,n,\epsilon')$-\CSH\ instance above, for $\epsilon' < \min\{\epsilon/11,1/40\}$, a partition with $n$ cuts corresponds to a solution to $\epsilon$-\GC.
\end{lemma}
 
\begin{proof}
First observe that since all of the agents $var_i$ and $copy_i$ are constructed to have at least $3/4$ of their valuation on $[d_i,d_i +3]$ and $[d_i+3,d_i+6]$ respectively, there must be at least one cut in each one of those intervals in any $\epsilon'$-approximate solution to Consensus-halving (with $\epsilon' < 1/4$) and therefore any $\epsilon'$-approximate solution to Consensus-halving with $2N$ cuts must have exactly one cut in each interval. Furthermore, since the constructed instance consists of $2N$ agents, by \cite{simmons2003consensus}, such a partition with $2N$ cuts is guaranteed to exist. 
 

Now consider such a solution $\CCC$ to $(n,n,\epsilon')$-\CSH\ with $2N$ cuts. For each agent $var_i$ (and associated gate $G^{\tau}$, if any), since her valuation in $v_i^{-}$ is at least the same as her valuation outside the interval $[d_i,d_i+3]$, the cut from $\CCC$ in $[d_i,d_i+3]$ must be in $[d_i+1-\epsilon',d_i+2+\epsilon']$, since $\CCC$ is a solution to $(n,n,\epsilon')$-\CSH. We will assume without loss of generality that the leftmost piece of the partition $\CCC$ is in $O_{-}$. Notice then that for each node $v_i$, the piece on the left-hand side of the cut in $v_i^{-}$ is always in $O_{-}$ and the piece on the left-hand side of the cut in $v_i^{+}$ is always in $O_{+}$. 
Let the location of the cut be $d_i+1+t_i^{-}$ where $t_i^{-} \in [-\epsilon',1+\epsilon']$. Analogously, the same argument holds for agent $copy_i$ and the interval $[d_i+3-\epsilon',d_i+4+\epsilon']$, and define $t_i^{+} \in [-\epsilon',1+\epsilon']$ similarly.

Now consider the agent $copy_i$ and the cut at location $d_i+1+t_i^{-}$. If $t_i^{-} \in [0,1]$, then since agent $copy_i$ has valuation $1$ on interval $v_i^{-}$, $t_i^{-}$ of her valuation will be on a piece in $O_{-}$ and $1-t_i^{-}$ of her valuation will be on a piece in $O_{+}$. Then, since $\CCC$ is a solution to $(n,n,\epsilon')$-\CSH, the cut in $d_i+3+t_i^{+}$ must be placed so that $|t_i^{-}-t_i^{+}| \leq \epsilon'/2$; similarly for the cases where $t_i^{-} \notin [0,1]$. In other words, $copy_i$ ensures that the cut at $d_i+1+t_i^{-}$ is ``copied" $\epsilon'$-approximately. 

We will interpret the solution $\CCC$ as a solution to $\epsilon$-\GC\ in the following way. For each node $v_i$ and each associated cut at $d_i+1+t_i^{-}$ let
\begin{ceqn}
\begin{equation}\label{eq:x}
x_i  := \begin{cases}
0				&,\ \ \ t_i^{-} < 0\\
t_i^{-}			&,\ \ \ t_i^{-} \in [0,1]\\
1				&,\ \ \ t_i^{-} > 1 \end{cases}
\end{equation}
\end{ceqn}
and notice 
\begin{ceqn}
\begin{equation}\label{eq:t}
|t_i^{+}  - x_i| \leq 2\epsilon' \ \ \ , \ \ \ |t_i^{-}  - x_i| \leq 2\epsilon'
\end{equation}
\end{ceqn}
To complete the proof, we just need to argue that these variables satisfy the constraints of the gates of the circuit. 

\medskip

\noindent \textbf{\emph{Constant gate}} $\tau = (G_{\zeta},nil,nil,v_{out},\alpha)$: The valuation of agent $var_{out}$ for the intervals $[d_i,d_i+1+\alpha]$ and $[d_i+1+\alpha,d_i+3]$ is the same and since the height of  the agent's value density function is at least 1 in $[d_i,d_i+3]$,\footnote{Notice that the constant gate is the only gate where $border_i$ and $G^\tau$ overlap.} it holds that $t_{out}^{-} \in [\alpha - \epsilon', \alpha + \epsilon']$. Then, by Equation \ref{eq:t}, it holds that $x_{out} \in [\alpha - 3\epsilon', \alpha + 3\epsilon']$, so for $\epsilon' < \epsilon/3$ the gate constraint is satisfied.

\medskip

\noindent \emph{\textbf{Multiplication-by-scalar gate}} $\tau = (G_{\times \zeta},v_{in},nil,v_{out},\alpha)$. Notice that for any given cut $t_{in}^{+}$ and $1-\alpha \geq \epsilon$, it holds that $t_{out}^{-} \in [\alpha t_{in}^{+} + \epsilon/2 - \epsilon',\alpha t_{in}^{+} + \epsilon/2 + \epsilon']$ as the height of $G^\tau$ in $v_{out}^{-}$ is at least 1. Similarly, for the case $1-\alpha < \epsilon$ and any given cut $t_{in}^{+}$, it holds that $t_{out}^{-}\in [\alpha t_{in}^{+} + (1-\alpha)/2 - \epsilon',\alpha t_{in}^{+} + (1-\alpha)/2 + \epsilon']$ as the height of $G^\tau$ in $v_{out}^{-}$ is at least 1. In particular, since $1-\alpha < \epsilon$, it also holds that $t_{out}^{-} \in [\alpha t_{in}^{+} + \epsilon/2 - \epsilon',\alpha t_{in}^{+} + \epsilon/2 + \epsilon']$ for this case as well. Then, by Equation \ref{eq:t}, it holds that $x_{out} \in [\alpha t_{in}^{+} + \epsilon/2 - 3\epsilon',\alpha t_{in}^{+} + \epsilon/2 + 3\epsilon']$ and since $\alpha \leq 1$ it also holds that $x_{out} \in [\alpha x_{in} + \epsilon/2 - 5\epsilon',\alpha x_{in} + \epsilon/2 + 5\epsilon']$, again by Equation \ref{eq:t}. Then the gate constraint is satisfied whenever $\epsilon' < \epsilon/10$.

\medskip

\noindent \emph{\textbf{Addition gate}} $\tau = (G_{{+}},v_{in_1},v_{in_2},v_{out},nil)$. If for the cuts $t_{in_1}^+$ and $t_{in_2}^+$ it holds that $t_{in_1}^{+} + t_{in_2}^{+} < 1 - \epsilon + 4\epsilon'$ then $t_{out}^{-}  \in [t_{in_1}^{+} + t_{in_2}^{+}-5\epsilon',t_{in_1}^{+} + t_{in_2}^{+} + 5\epsilon']$ as the height of $G^\tau$ in $v_{out}^{-}$ is at least 1. This then implies that $x_{out} \in [x_{in_1}^{+} + x_{in_2}^{+}-11\epsilon',x_{in_1}^{+} + x_{in_2}^{+} + 11\epsilon']$, by Inequality \ref{eq:t}. On the other hand, when $t_{in_1}^{+} + t_{in_2}^{+} \geq 1 - \epsilon + 4\epsilon'$, then by Definition \ref{eq:x}, it holds that $x_{in_1}+x_{in_2} \in [1-\epsilon,1]$ and clearly $t_{out}^{-} \in [1 - \epsilon, 1+\epsilon']$ which by Definition \ref{eq:x} implies that $x_{out} \in [1-\epsilon,1]$. The gate constraints are satisfied for $\epsilon' < \epsilon/11$ for each of the cases.

\medskip


\noindent \emph{\textbf{Subtraction gate}} $\tau = (G_{{-}},v_{in_1},v_{in_2},v_{out},nil)$. Analogously to the addition gate described above, when for the cuts $t_{in_1}^{+}$ and $t_{in_1}^{+}$ it holds that $t_{in_1}^{+} - t_{in_2}^{+} > \epsilon - 4\epsilon'$ then $t_{out}^{-} \in [t_{in_1}^{+} - t_{in_2}^{+}-5\epsilon',t_{in_1}^{+} - t_{in_2}^{-} + 5\epsilon']$ as the height of $G^\tau$ in $v_{out}^{-}$ is at least 1. This implies that $x_{out} \in [x_{in_1}^{+} - x_{in_2}^{-}-11\epsilon',x_{in_1}^{+} - x_{in_2}^{-} + 11\epsilon']$ by Inequality \ref{eq:t}. On the other hand when $t_{in_1}^{+} - t_{in_2}^{-} \leq \epsilon - \epsilon'$, which implies that $x_{in_1}+x_{in_2} \in [0,\epsilon]$ by Definition \ref{eq:x}, it clearly holds that $t_{out}^{-} \in [-\epsilon',\epsilon]$ and hence by Definition \ref{eq:x}, we have $x_{out} \in [0,\epsilon]$. The gate constraints are satisfied for $\epsilon' < \epsilon/11$ for each of the cases.

\medskip

\noindent \emph{\textbf{Less-than-equal gate}} $\tau = (G_{{<}},v_{in_1},v_{in_2},v_{out},nil)$. We will consider three cases, depending on the positions of the cuts $t_{in_1}^{+}$ and $t_{in_2}^{-}$. First, when $|t_{in_1}^{+} - t_{in_2}^{-}| < \epsilon- 4\epsilon'$, by Inequality \ref{eq:t} it holds that $|x_{in_1} - x_{in_2}| < \epsilon$ and the output of the gate is unconstrained. When $t_{in_1}^{+} - t_{in_2}^{-} \geq \epsilon - 4\epsilon'$ then by Inequality \ref{eq:t} it holds that $x_{in_1} \geq  x_{in_2}+\epsilon$. Additionally, since the height of $G^\tau$ in $[v_{out,r}^{-} - \epsilon,v_{out,r}^{-}]$ is at least 1, it holds that $t_{out}^{-} \in [1 - \epsilon,1+\epsilon']$, which by Definition \ref{eq:x} implies that $x_{out}^{-} \in [1 - \epsilon,1]$ and the gate constraint is satisfied. The argument for the case when $t_{in_2}^{-} > t_{in_1}^{+} - 2\epsilon'$ is completely symmetrical. 

\medskip

\noindent \emph{\textbf{Logic OR gate}} $\tau = (G_{{\vee}},v_{in_1},v_{in_2},v_{out},nil)$. We will consider three cases depending on the position of the cuts $t_{in_1}^{+}$ and $t_{in_2}^{+}$. First, when $t_{in_1}^{+} + t_{in_2}^{+} < 0.4$ it holds that $t_{out}^{-} \in [-\epsilon', \epsilon]$ and hence by Definition \ref{eq:x}, it holds that $x_{out} \in [0,\epsilon$]. Furthermore, by Inequality \ref{eq:t} it holds that $x_{in_1} + x_{in_2} < 0.4 + 4\epsilon'$ and for $\epsilon' < 1/40$, it also holds that $x_{in_1},x_{in_2} < 0.5$ and the gate constraint is satisfied. Next, when $t_{in_1}^{+} + t_{in_2}^{+} \in [0.4,0.8]$ then by Inequality \ref{eq:t}, it holds that $x_{in_1},x_{in_2}  \in [0.4-\epsilon',0.8+4\epsilon']$ and in particular, when $\epsilon'<1/40$ then it also holds that $x_{in_1}+x_{in_2} \in [0.3,0.9]$ and the output of the gate in unconstrained. Finally when $t_{in_1}^{+} + t_{in_2}^{+} > 0.8$, it holds that $t_{out}^{-} \in [1 - \epsilon, 1 + \epsilon']$ and hence by Definition \ref{eq:x}, we have that $x_{out} \in [1 -\epsilon,1]$. Furthermore, by Inequality \ref{eq:t} we have that $x_{in_1}+x_{in_2} > 0.8 + 4\epsilon'$ which is greater than $0.9$ when $\epsilon' < 1/40$ which implies that at least one of the two inputs is greater than $\epsilon$. In particular, the gate's output lies in $[1-\epsilon,\epsilon]$ when the inputs are smaller than $\epsilon$ or greater than $1-\epsilon$ and at least one of them is greater than $1-\epsilon$. This shows that the gate constraint is satisfied.

\medskip

\noindent \emph{\textbf{Logic AND gate}} $\tau = (G_{{\wedge}},v_{in_1},v_{in_2},v_{out},nil)$. Analogously to the Logical OR gate, we will consider three cases, depending on the position of the cuts $t_{in_1}^{+}$ and $t_{in_2}^{+}$. First, when $t_{in_1}^{+} + t_{in_2}^{+} < 1.4$, it holds that $t_{out}^{-} \in [-\epsilon',\epsilon]$ and by Inequality \ref{eq:x}, we have that $x_{out} \in [0,\epsilon]$. On the other hand, by Inequality \ref{eq:t}, we have that $x_{in_1}+ x_{in_2} < 1.4 + 4\epsilon'$ which is at most $1.5$ for $\epsilon' < 1/40$ and the gate constraint is satisfied for the same reason as in the third case of the Logic OR gate above (using the argument symmetrically, for values smaller than $\epsilon$ instead of at larger than $1-\epsilon$). Next, when $1.4 \leq t_{in_1}^{+} + t_{in_2}^{+} \leq 1.8$, by Inequality \ref{eq:t} and for $\epsilon' < 1/40$, it holds that $x_{in_1}+ x_{in_2} \in [1.3,1.9]$ and the output of the gate is unconstrained. Finally, when $t_{in_1}^{+} + t_{in_2}^{+} > 1.8$, by Inequality \ref{eq:t} and for $\epsilon' < 1/40$, it holds that $x_{in_1}+x_{in_2} > 1.7$. Furthermore, it holds that $t_{out}^{-} \in [1-\epsilon,1+\epsilon']$ and hence by Definition \ref{eq:x}, we have that $x_{out} \in [1-\epsilon,1]$ and the gate constraint is satisfied.

\medskip

\noindent \emph{\textbf{Logic NOT gate}} $\tau = (G_{{\neg}},v_{in},v_{in_2},v_{out},nil)$. We will consider three cases, depending on the location of the cut $t_{in}^{-}$. First, if $t_{in}^{-} < 0.4$, it holds that $t_{out}^{-} \in [-\epsilon',\epsilon]$ and hence by Definition \ref{eq:x}, we have that $x_{out} \in [0,\epsilon]$. Furthermore, by Inequality \ref{eq:t}, it holds that $x_{in} < 0.4 + 4\epsilon'$ which is at most $0.5$ when $\epsilon' < 1/40$ and the gate constraint is satisfied because for any value of the input smaller than $\epsilon$, the output is in $[0,\epsilon]$. Next, when $t_{in}^{-} \in [0.4,0.8]$ and for $\epsilon < 1/40$, by Inequality \ref{eq:t} it holds that $x_{in} \in [0.3,0.7]$ and the output of the gate in unconstrained. Finally, when $t_{in}^{-} > 0.8$, it holds that $t_{out}^{-} \in [1-\epsilon,1+\epsilon']$ and by Definition \ref{eq:x}, we have that $x_{out}  \in [1-\epsilon,1]$. Furthermore, by Inequality \ref{eq:t} and for $\epsilon' < 1/40$, it holds that $x_{in}>0.9$ and the gate constraint is satisfied for a reason analogous to the one described above. 

\medskip

\noindent Given the discussion above, by setting $\epsilon' < \min\{\epsilon/11,1/40\}$\footnote{We can in fact assume some $\epsilon \leq 11/40$, as the smaller the $\epsilon$, the harder the problem is, since we are interested in establishing hardness for some constant $\epsilon$.}, the gate constraints are satisfied, and the vector $(x_i)$ obtained from $\CCC$ is a solution to $\epsilon$-\GC.
\end{proof}

\noindent Now from Lemma \ref{lem:reduc}, we obtain the following result.

\begin{theorem}\label{thm:PPADhard}
There exists a constant $\epsilon'>0$ such that $(n,n,\epsilon')$-\CSH\ is PPAD-hard.
\end{theorem}
\begin{proof}
Recall that in the proof of Lemma \ref{lem:reduc}, $\epsilon'$ was constrained to be at most $\min\{1/40,\epsilon/11\}$ and in particular by Lemma \ref{lem:gcppad}, there exists a constant $\epsilon'$ that would make the reduction work. Recall however that we ``expanded'' the instance of $(n,\epsilon')$-{\sc{Chalving}} from the interval $[a,b]$ to $[0,6N]$ for convenience, which implies that after rescaling the instance to a constant interval $[a,b]$, the allowed error $\epsilon'$ goes down to $O(1/n)$. To get a constant error $\epsilon'$, we simply multiply all valuations by $N$. 
\end{proof}

\noindent Theorem \ref{thm:PPADhard} implies that although a solution with $n$ cuts is generally desirable, it might be hard to compute, even for a relatively simple class of valuations like those used in the reduction. In fact, we can extend our results to the more general case of finding a partition with $n+k$ cuts where $k$ is a constant.

\begin{theorem}\label{thm:PPADhardconstant}
Let $k$ be any constant. Then there exists a constant $\epsilon'$ such that $(n,n+k,\epsilon')$-{\sc{Chalving}} is PPAD-hard.
\end{theorem}


\begin{proof}
Let $S=(V,\mathcal{T})$ be an instance of $\epsilon$-\GC\ with $N$ nodes, consisting of smaller identical sub-circuits $S_i=(V_i,\mathcal{T}_i)$, for $i=1,2,\ldots,k+1$, with with $N/(k+1)$ nodes each such that for all $i,j \in [k+1]$ such that $i \neq j$, it holds that $V_i \cap V_j = \emptyset$. and $\mathcal{T}_i \cap \mathcal{T}_j = \emptyset$. In other words, the circuit $S$ consists of $k+1$ copies of a smaller circuit $S_i$ that do not have any common nodes or gates. Furthermore, for convenience, assume without loss of generality that for two nodes $l$ and $m$ such that $u_l \in V_i$ and $u_m \in V_j$, with $i < j$, it holds that $l < m$. In other words, the labeling of the nodes is such that nodes in circuits with smaller indices have smaller indices.

Let $H$ be the instance of $(n,n,\epsilon')$-\CSH\ corresponding to the circuit $S$ following the reduction described in the beginning of the section and recall that $n=2N$ in the construction. Note that according to the convention adopted above for the labeling of the nodes, for $i<j$, the agents corresponding to $V_i$ lie in the interval $[\ell_i,r_i]$, whereas the agents corresponding to $V_j$ lie in the interval $[\ell_j, r_j]$ and $r_i \leq \ell_j$. In other words, agents corresponding to sub-circuits with smaller indices are placed before agents with higher indices, and there is no overlap between agents corresponding to different sub-circuits.

Now suppose that we have a solution to $(n,n+k,\epsilon')$-\CSH. Since there is no overlap between valuations corresponding to different sub-circuits, an approximate solution with $n+k$ cuts for the instance $H$ implies that there exists some interval $[\ell_i,r_i]$ corresponding to the set of nodes $V_i$ of sub-circuit $S_i$, such that at least $n/(k+1)$ cuts lie in $[\ell_i,r_i]$, otherwise the total number of cuts on $H$ would be at least $n+k+1$. Since there are exactly $n/(k+1)$ agents with valuations on $[\ell_i,r_i]$, this would imply an approximate solution for $n'$ agents with $n'$ cuts and the problem reduces to $(n,n,\epsilon')$-\CSH. 
\end{proof}

\section{$(n,n-1,\epsilon)$-{\sc{Chalving}} is NP-hard}\label{sec:NPhard}

In the previous section, we proved that the problem of finding an approximate solution with $n$ players and $n$ cuts is PPAD-complete. Recall that for that case, we know that a solution exists \cite{simmons2003consensus}. For $n$ players and $n-1$ cuts however, we don't have the same guarantee. We prove that deciding whether this is the case or not is NP-hard.

\begin{theorem}
	There exists a constant $\epsilon'>0$ such that $(n,n-1,\epsilon')$-\CSH\ is NP-hard.
\end{theorem}

\begin{proof}
We will first describe the construction that we will use in the reduction. For consistency with the previous section, we will denote the error of the Consensus-halving problem by $\epsilon'$ and the error of the (implict) generalized circuits by functions of $\epsilon$. Let $R_{\epsilon}(S)$ be the construction for the reduction of Section \ref{sec:PPADhard} that encodes an $\epsilon$-generalized circuit $S$ into an $(n,n-1,\epsilon')$-\CSH\ halving instance when $\epsilon' < \epsilon/11$. We will reduce from 3-SAT, which is known to be NP-complete.

\noindent Let $\phi$ be any 3-SAT formula with $m$ clauses, $k\leq 3m$ variables $x_1,\ldots,x_k$, and let $\epsilon>0$ be given. For convenience of notation, let $\delta = \epsilon/11$. We will (implicitly) create a generalized circuit $S$ with the following building blocks:
\medskip
\begin{itemize}
	\item $k$ input nodes $x_1, \ldots, x_k$ corresponding to the variables $x_1,\ldots,x_k$.
	
	\medskip
	
	\item $k$ sub-circuits $\text{Bool}(x_i)$ for $i=1,2,\ldots,k$ that input the real value $x_i \in [0,1]$ and output a boolean value $x_i^{bool} \in \left[0,4\delta\right] \cup \left[1-4\delta,1\right]$ (see the lower stage of Figure \ref{fig:3SAT}). The allowed error for these circuits will be $\delta$. The implementation of the circuit in terms of the gates of the generalized circuit can be seen in Algorithm \ref{alg:bool}. Note that the sub-circuit containing all the $\text{Bool}(x_i)$ sub-circuits has at most $4k$ nodes as each $\text{Bool}(x_i)$ sub-circuit could be implemented with one constant gate, one subtraction gate, one addition gate and one equality gate; the latter is to maintain the convention that all inputs to each gate are distinct.
	
	\medskip
	
	\item A sub-circuit $\Phi(x_1^{bool},\ldots,x_k^{bool})$ that implements the formula $\phi$, inputing the boolean variables $x_i^{bool}$ and outputting a value $x_{out}$ corresponding to the value of the assignment plus the error introduced by the approximate gates. The allowed error for this circuit will be $4\delta$. A pictorial representation of such a circuit can be seen in Figure \ref{fig:3SAT}; note that the circuits $\text{Bool}(x_i)$ are also shown in the picture. This circuit has at most $k+3m$ nodes. First, there might be $k$ possible negation gates to negate the input variables. Secondly, for each clause, in order to implement an $OR$ gate of fan-in $3$, we need $2$ $OR$ gates of fan-in $2$, for a total of $2m$ gates for all clauses. Finally, in order to simulate the $AND$ gate with fan-in $m$, we need $m$ $AND$ gates of fan-in $2$. Overall, since $k \leq 3m$, we need at most $6m$ nodes to implement this sub-circuit, using elements of the generalized circuit.
	
	\medskip
		
	\item A sub-circuit $\text{Rebool}(x_1,\ldots,x_k,x_{out})$ that inputs the variables $x_i$, for $i=1,2,\ldots,k$ and the variable $x_{out}$ and computes the function 
	\begin{ceqn}
	\[
	\min ( x_{out} , \max(x_1, 1-x_1) , \ldots , \max(x_k, 1-x_k)).\]
	\end{ceqn}
	 \noindent The function can be computed using the gates of the generalized circuit as shown in Algorithm \ref{alg:maxmin}. Let $x_{out}^{bool}$ be the output of that sub-circuit with allowed error  $4\delta$. Note that this circuit has at most $16k$ nodes. Each min and max operation requires $8$ nodes and we need to do $2k$ such computations overall; $k$ for the $k$ $\max$ operations and $k$ to implement the $\min$ operation of fan-in $k$ with $\min$ operations of fan-in $2$. Again, since $k \leq 3m$, this sub-circuit requires at most $48m$ nodes in total.
\end{itemize}

\noindent Following the notation introduced above, let $R_{\delta}(\text{Bool})$, $R_{4\delta}(\Phi)$ and $R_{4\delta}(\text{Rebool})$ denote the valuations of the agents in the instance of Consensus-halving corresponding to those sub-circuits, according to the reduction described in Section \ref{sec:PPADhard}. In other words, based on the circuit described above, we create an instance $H$ of Consensus-halving where we have:
\medskip
\begin{itemize} 
	\item $2k$ agents (as each node corresponds to two agents, $var_i$ and $copy_i$) that correspond to the input variables $x_1,\ldots,x_k$, who are not the output of any gate,
	\medskip
	\item at most $2(4k+k+3m+16k)$ nodes corresponding to the internal nodes and the output node of the circuit,
	\medskip
	\item an additional agent with valuation
	\begin{ceqn}
\[
u_n  = \begin{cases}
	1				, \text{if } t \in [b - 18m\epsilon' - 1,b - 18m\epsilon'] \\ 
	1				, \text{if } t \in [b,b+1]\\
	0				, \text{otherwise} \end{cases}\\ \]
\end{ceqn}
\noindent where $[a,b]$ is the interval where the value of $x_{out}^{bool}$ is ``read'' in the instance of Consensus-halving, i.e. the interval where the cut $t_{out}^{bool}{-}$ will be placed in the Consensus-halving solution. 
\end{itemize}	

\noindent Recall Definition \ref{eq:x} from Section \ref{sec:PPADhard} and note that as far as agent $n$ is concerned, any cut $t_{out}^{bool}{-}$ such that $1 - 18m\epsilon \leq x_{out}^{bool} \leq 1$ is a Consensus-halving solution.\\


	\begin{algorithm}[b]
		\caption{Computing $bool(x)$.}
		\label{alg:bool}
		\begin{algorithmic}
			\State $a \gets x - 1/4$
			\State $bool \gets a + a$
		\end{algorithmic}
	\end{algorithm}

\begin{algorithm}[t]
	\caption{Computing $\min(x,y)$ and $\max(x,y)$.}
	\label{alg:maxmin}
	\begin{algorithmic}
		\State $a \gets x-y$\ ; \ \ $b \gets y-x$\ ; \ \ $c\gets a+b$
		\State $d \gets c/2$\ ; \ \ $\ell \gets (x/2)+(y/2)$
		\State $min \gets \ell-d$\ ; \ \ $max \gets \ell+d$
	\end{algorithmic}
\end{algorithm}

\begin{figure}[t!]
  \centering
    \begin{tikzpicture}[scale=0.7]
\draw (0,0) rectangle (1,1) node[pos=.5] {\small{Bool}};
\draw (1.5,0) rectangle (2.5,1) node[pos=.5] {\small{Bool}};
\draw (3,0) rectangle (4,1) node[pos=.5] {\small{Bool}};

\draw (3,1.5) rectangle (4,2.5) node[pos=.5] {$\lnot$};

\draw (1,2) rectangle (2,3) node[pos=.5] {$\vee$};

\draw (3,4) rectangle (4,5) node[pos=.5] {$\vee$};

\draw (6,4) rectangle (7,5) node[pos=.5] {$\vee$};

\draw (4.5,6) rectangle (5.5,7) node[pos=.5] {$\wedge$};

\fill (5.5,2) circle (0.1cm);
\fill (6.5,2) circle (0.1cm);
\fill (7.5,2) circle (0.1cm);

\node (x1) at (0.5,-0.5) {$x_1$};
\node (x2) at (2,-0.5) {$x_2$};
\node (x3) at (3.5,-0.5) {$x_3$};
\draw (x1) -- (0.5,0);
\draw (x2) -- (2,0);
\draw (x3) -- (3.5,0);

\draw (0.5,1) -- (1.5,2);
\draw (2,1) -- (1.5,2);

\draw (3.5,1) -- (3.5,1.5);

\draw (1.5,3) -- (3.5,4);

\draw (3.5,2.5) -- (3.5,4);

\draw (3.5,5) -- (5,6);

\draw (6.5,5) -- (5,6);

\node (xout) at (5,7.5) {$x_{out}$};

\draw (xout) -- (5,7);

\draw (6.5,4) -- (5.5,3);
\draw (6.5,4) -- (7.5,3);


\end{tikzpicture}
    \caption{A generalized circuit corresponding to a $3SAT$ formula $\phi$, where the first clause is $(x_1 \vee x_2 \vee \overline{x_3})$. The nodes of the circuit between different layers are omitted. The layer at the output layer that ``restores'' the boolean values is also not shown, therefore $x_{out}$ is the outcome of the emulated formula $\phi$.
    \label{fig:3SAT}}
\end{figure}
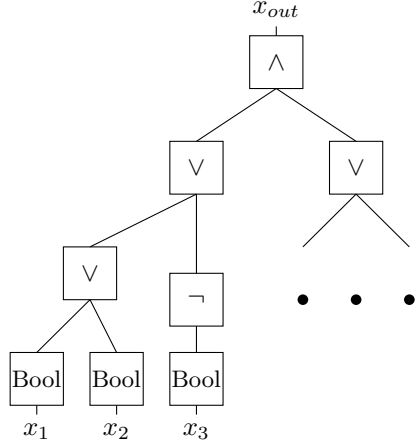

\noindent We will now argue about the correctness of the reduction. Let $n$ be the number of agents and notice that there are $n-1$ agents that correspond to the nodes of the circuit and a single agent constraining the value of $x_{out}^{bool}$. Notice that since the allowed error for the sub-circuit $\text{Rebool}(x_1,\ldots,x_k,x_{out})$ is $4\delta$, the total additive error of the agents of $R_{4\delta}(\text{Rebool})$ will be at most $4\delta \cdot 48m \leq 18m\epsilon'$.

First, assume that there exists a a solution to $\epsilon'$-approximate Consensus-halving with $n-1$ cuts. By the correctness of the construction of Section \ref{sec:PPADhard} and the fact that $\epsilon' < \epsilon/11=\delta$, the solution encodes a valid assignment to the variables of the generalized circuit $S$. Due to the valuation of agent $n$, the output of $C$ must satisfy 

\begin{ceqn}
\[x_{out}^{bool} \geq 1 - 18m\epsilon' - \epsilon',\]
\end{ceqn}
otherwise the corresponding cut $t_{out}^{bool}{-}$ could not be a part of a valid solution. Since the total additive error for the circuit $\text{Rebool}(x_1,\ldots,x_k,x_{out})$ is at most $18m\epsilon'$, if we choose $\epsilon' < 1/90m$, it holds that 

\begin{ceqn}
\[x_{out}^{bool} \geq 4/5 - \epsilon'  \ \ \text{which implies that} \ \ x_{out} \geq 3/4,\]
\end{ceqn}
 by the function implemented by the circuit $\text{Rebool}(x_1,\ldots,x_k,x_{out})$. For the same reason, for each $i=1,\ldots,k$ it holds that 
 
 \begin{ceqn}
 \[x_i \in [0,1/4] \cup [3/4,1]\] 
 \end{ceqn}
 and hence the output of $\text{Bool}(x_i)$ will lie in $[0,4\delta] \cup [1-4\delta,1]$, 
which means that the inputs $x_1^{bool},\ldots,x_k^{bool}$ to the gates of the sub-circuit $\Phi(x_1^{bool},\ldots,x_k^{bool})$ will be treated correctly as boolean values by the gates of the circuit (since the allowed error of the sub-circuit is $4\delta$). Since the circuit $\Phi(x_1^{bool},\ldots,x_k^{bool})$ computes the boolean operations correctly and $x_{out} \geq 3/4$, the formula $\phi$ is satisfiable.
   

\noindent For the other direction, assume that $\phi$ is satisfiable and let $\tilde{x}=(\tilde{x}_1,\ldots,\tilde{x}_k)$ be a satisfying assignment. First we set the values of the  variables $x_1, ..., x_k$ to $0$ or $1$ according to $\tilde{x}$ and then we propagate the values up the circuit $S$ using the exact operation of the gates, which by our construction can be encoded to an instance of exact Consensus Halving for the $(n-1)$ agents corresponding to the nodes of $S$, i.e. the first $n-1$ will be exactly satisfied with the partition resulting from the encoded satisfying assignment. For the $n$-th agent, again, since the total additive error is bounded by $18m\epsilon'$, agent will be satisfied with the solution.
\end{proof}
Again, we remark here that by using similar arguments as in the proof of Corollary \ref{col:PPADprob}, we can prove that the result holds even when the valuations are probability measures. 
\section{$(n,n,\epsilon)$-{\CSH} is in PPA}\label{sec:inPPA}

In this section, we prove that $(n,n,\epsilon)$-\CSH\ is in PPA. As we discussed in the introduction, this result of ours was referenced in \cite{filos2018consensus} to complement the PPA-hardness reduction of the inverse-exponential precision version and obtain PPA-completeness. For establishing this result, we construct a reduction from $(n,n,\epsilon)$-\CSH\ to the PPA-complete problem {\sc Leaf} which goes via $(n,T)$-\tucker, the computational version of Tucker's Lemma. 

\begin{theorem}
	$(n,n,\epsilon)$-\CSH\ is in PPA.
\end{theorem}

\begin{proof}
	The result is a corollary of Lemma \ref{thm:inPPA} and Lemma \ref{lem:tuckerppa}, which will be proved in the remainder of the section.
\end{proof}	
Before we proceed, we provide some intuition. In \cite{simmons2003consensus}, Simmons and Su designed an algorithm for finding an approximate solution to the Consensus-Halving problem given access to an algorithm that solves Tucker on a triangulated cross-polytope (the formal definitions are given below). In the proof of Lemma \ref{thm:inPPA} we use this algorithm directly to reduce the computational version of Consensus Halving, $(n,n,\epsilon)$-\CSH, to the computational version of Tucker, $(n,T)$-\tucker, defined below. 

Then, the inclusion of $(n,n,\epsilon)$-\CSH\ in PPA will follow from the the fact that $(n,T)$-\tucker\ is in PPA. This was proven in \cite{papadimitriou1994complexity} where the computational version of the problem is defined on a subdivision of the hypercube, rather than the triangulation of the cross-polytope, as required for the Simmons-Su algorithm \cite{simmons2003consensus}. For the latter problem, a proof was sketched in \cite{papadimitriou1994complexity}, where the idea was to map a triangulated hemisphere continuously and symmetrically to the hypercube. We provide a formal proof here, but following a different approach;\footnote{A map between the two topological spaces or the corresponding solutions seems like the most natural approach, but formalizing the intuition seems like it could result in significant technical clutter. Our approach ``moves'' the technical difficulty to the definition of the ``aligned with hemispheres'' triangulation, which can rather be easily understood intuitively. Then, the two lemmas that establish the result are basically adaptations of the known algorithms of \cite{simmons2003consensus} (Lemma \ref{thm:inPPA}) and \cite{prescott2005constructive} (Lemma \ref{lem:tuckerppa}).} we use an algorithm proposed by Prescott and Su \cite{prescott2005constructive} (for proving a generalization of Tucker's Lemma due to Fan \cite{fan1967simplicial}) and prove that it can be used to recover a solution to $(n,T)$-\tucker\ from a solution to \textsc{Leaf}. The algorithm requires $(n,T)$-\tucker\ to be defined on a specific type of triangulation $T$, which we also define below.

\subsection{The Borsulk-Ulam Theorem and Tucker's Lemma}

Denote the $n$-dimensional Euclidean space by $\RR^n$. The $(n+1)$-dimensional unit ball is $$B^{n+1}=\{ {\bf x}\in \RR^{n+1} : \sum_{i=1}^{n+1} |x_i|^2 \le 1 \}$$ and its surface is the $n$-dimensional sphere $$S^{n}=\{ {\bf x}\in \RR^{n+1} : \sum_{i=1}^{n+1} |x_i|^2 = 1 \}.$$ The $(n+1)$-dimensional cross-polytope is $$P^{n+1}=\{ {\bf x}\in \RR^{n+1} : \sum_{i=1}^{n+1} |x_i| \le 1 \},$$ i.e., unit ball in the $l_1$-norm. Denote the surface of $P^{n+1}$ by $$C^n=\{ {\bf x}\in \RR^{n+1} : \sum_{i=1}^{n+1} |x_i| = 1 \}.$$ The Borsuk-Ulam theorem is the following:

\begin{theorem}[Borsuk-Ulam theorem]\label{thm:bu}
Any continuous function $f : C^n \mapsto \RR^n$ must have a vertex $\mathbf x \in C^n$, such that $f({\bf x})=f(-{\bf x})$.
\end{theorem}
Note that although the theorem is originally stated on  domain $S^n$, it is also true when the domain of $f$ is the surface of the cross-polytope $C^{n}$.

In general, a \textit{subdivision} or \textit{simplicization} is a partition of a geometric object into small objects such that any two such small objects either share a common facet or do not intersect. In particular, we have the following definitions:

\begin{definition}[Triangulation] \cite{simmons2003consensus}\label{def:triangulation}
A \textit{triangulation} $T$ of a geometric object $X$ is a collection of (distinct) $n$-simplices $\sigma_1, \ldots, \sigma_m$ whose union is $X$ such that for all $i$ and $j$, $\sigma_i \cap \sigma_j$ either contains a common facet of $\sigma_i$ and $\sigma_j$ or a lower dimension face, or is empty. 
\end{definition}
The \emph{mesh size} $\tau$ of a triangulation $T$ is the maximum distance between any two vertices of $T$.

\begin{definition}[Centrally symmetric and $d$-skeleton] \cite{prescott2005constructive}
	A triangulation $T$ of $S^{n}$ is \textit{centrally symmetric}  if given simplex $\sigma \in T \cap S^n$, then $-\sigma \in T \cap S^n$, where $-\sigma$ is the simplex obtained by negating each vertices of $\sigma$. The notion can also be defined on other centrally symmetric objects, e.g., on $C^n$. The $d$-\emph{skeleton} of a triangulation $T$ is the collection of simplices of $T$ of dimension at most $d$, i.e. $\{ \sigma \in T: \text{dim}(\sigma)\leq d\}$.
\end{definition}	

\noindent Tucker's Lemma can be formulated in a number of different ways; one is the following.

\begin{lemma}[Tucker's Lemma \cite{simmons2003consensus}]\label{lem:tucker}
Let $T$ be an centrally symmetric triangulation of $C^{n}$ whose vertices are assigned labels from $\{+1, -1, +2, -2, \ldots,+n, -n\}$. The labels of antipodal vertices sum to zero, i.e., the labelling function $\lambda$ satisfies $\lambda(-{\bf x})=\lambda({\bf x})$ for any vertex ${\bf x} \in C^n$. Then there must exist a 1-simplex (which is referred to as a complementary edge) such that its two vertices have labels that sum to zero.
\end{lemma}
As pointed out in the literature \cite{simmons2003consensus,prescott2005constructive}, the lemma is often stated for a triangulation of a ball, but for it also holds for a triangulation of a sphere (obtained by gluing two $n$-balls along their boundaries, see \cite{simmons2003consensus} or \cite{prescott2005constructive} for more details). We are interested in triangulations that satisfy the following property.

\begin{definition}[Aligned with hemispheres Triangulation]\cite{prescott2005constructive}
	A {\em flag of hemispheres} in $S^n$ is a sequence $H_{0} \subset \dots \subset H_{n}$ where each $H_d$ is homeomorphic to a $d$-ball, and for $1\le d \le n$, $\partial H_d = \partial (-H_d)=H_d \cap (-H_d)=H_{d-1} \cup (-H_{d-1}) \approxeq S^{d-1}$, $H_n \cap (-H_n)=S^n$, and $H_0, -H_0$ are antipodal points. A symmetric triangulation $T$ of $S^n$ is said to be {\em aligned with hemispheres} if we can find a flag of hemispheres such that $H_d$ is contained in the $d$-skeleton of the triangulation. 
\end{definition}
Note that the same definition can be adapted to our context by converting $l_2$-norm to $l_1$-norm.
We can now define the computational version of the Tucker problem.
\begin{problem}{$(n,T)$-{\tucker}}. Let $T$ be a fixed centrally symmetric triangulation of mesh size $\tau$ (that is aligned with hemispheres). Given an integer $n$ and a polynomial-size circuit $\lambda$ that computes for each vertex $x$ on $T$ a label $\lambda: T \rightarrow\{+1, -1, \cdots, +n, -n\}$ such that $\lambda(-{\bf x}) + \lambda({\bf x}) = 0$, find two adjacent vertices ${\bf z}$, ${\bf z}'$ of $T$, with $\lambda({\bf z}) + \lambda({\bf z}') = 0$.
\end{problem}	
By ``fixed'' here we mean that the circuit $\lambda$ knows how the vertices of the triangulation are represented as inputs and can produce a label for them. Such a triangulation can be constructed, e.g. see \cite{freund1981constructive}. Note that the triangulation $T$ is not given explicitly as input to the problem, but it is rather accessed via the labelling circuit, therefore solutions that exhaustively search over the vertices of $T$ for complementary edges are not efficient.

\subsection{From Consensus-halving to Tucker}
First, we establish the reduction from Consensus-Halving to Tucker. 

\begin{lemma}\label{thm:inPPA}
	$(n,n,\epsilon)$-\CSH\ reduces to $(n,T)$-\tucker\ in polynomial time, where $T$ is a symmetric triangulation $T$ of $C^n$ with mesh size at most $\epsilon/2M$ and $M$ is the upper bound on the valuation functions of $(n,n,\epsilon)$-\CSH.
\end{lemma}

\begin{proof}
	Given any instance of $(n,n,\epsilon)$-\CSH, we construct an instance of $(n,T)$-\tucker\ based on the construction in \cite{simmons2003consensus}. We note that the coordinates of any vertex ${\mathbf x} \in C^n$ naturally correspond to a partition that uses $n$ cuts on the $[0,1]$ interval.\footnote{The use of the $[0,1]$ interval is for convenience and without loss of generality; for any choice of the interval we could use a cross-polytope corresponding to a sphere of a different radius.} This is because the coordinates of any vertex ${\mathbf x} \in C^n$ satisfy $\sum_{i=1}^{n+1} |x_{i}|=1$, and a partition with $n$ cuts on $[0,1]$ can be interpreted as partitioning the interval into $n+1$ pieces such that the length of each piece is equal to $|x_i|, i=1,\ldots,n+1$. Furthermore, if the sign of the $i$-th coordinate $x_i$ is ``+'', piece $|x_i|$ is assigned to portion $O_{+}$; otherwise it is assigned to portion $O_{-}$. This interpertation is the basic bulding block of the Simmons-Su algorithm \cite{simmons2003consensus}.
	
	Let $T$ be a fixed triangulation of $C^n$, in the sense described above with mesh size $\tau=\epsilon/2M$, where $M$ is the bound of the valuation function on any subinterval. It is not hard to verify that for any two adjacent vertices $\mathbf{x}$ and $\mathbf{x}'$ (denote their associated portions by $O_{+}$, $O_{-}$, and $O'_{+}$,$O'_{-}$, respectively), it holds that
	\begin{ceqn}\[|v_i(O_{+}) - v_i(O'_{+})| \leq \epsilon /2, \ \ \text{for all}\ \  i \in [n]. \] \end{ceqn}
	To see this, note that two adjacent vertices in the triangulation $T$ can only differ by at most $\tau$ in distance, which means that the cuts corresponding to the two points $\mathbf{x}$ and $\mathbf{x}'$ will define intervals that are close to each other; in particular the intervals defining the portions $O_{+}$ and $O'_{+}$ will differ by at most $\tau$ in length. Since the difference in value for the two intervals could be at most $M$, the total difference in value could be at most $M\cdot \tau$ which is upper bounded by $\epsilon/2$, by the choice of $\tau$. Symmetrically, we also get that 
	\begin{ceqn}\[|v_i(O_{-}) - v_i(O'_{-})| \leq \epsilon /2, \ \ \text{for all}\ \  i \in [n]. \] \end{ceqn}
	
	We now label all the vertices on $T \cap C^n$. For any vertex $\mathbf x \in T \cap C^n$, denote 
	\begin{ceqn}\[l=\arg\max_i \{ |v_i(O_{+}) - v_i(O_{-})| \},\]\end{ceqn}
	and then label  $\mathbf x$ by $+l$ if $v_l(O_{+}) - v_l(O_{-}) > 0$;  label $\mathbf x$ by $-l$ if $v_l(O_{+}) - v_l(O_{-}) < 0$.
	Note that in case $v_l(O_{+}) - v_l(O_{-})= 0$ then $\mathbf x$ corresponds to a  solution of $(n,n,\epsilon)$-\CSH. We claim that this labeling satisfies the boundary condition of Lemma \ref{lem:tucker}. In summary, given an instance of $(n,n,\epsilon)$-\CSH, we have constructed an instance of $(n,T)$-\tucker.
	
	Now, given a solution to $(n,T)$-\tucker, i.e., two adjacent vertices $\bf z$ and ${\bf z}'$ that are assigned labels which sum to $0$,  assume without loss of generality that ${\bf z}$ (with associated portions $O_{+}$ and $O_{-}$) is labelled by $k$ and ${\bf z}'$ (with associated portions $O'_{+}$ and $O'_{-}$) is labelled by $-k$. According to the labelling, it holds that $v_k(O_{+}) - v_k(O_{-}) >0$ and $v_k(O'_{+}) - v_k(O'_{-}) <0$. Therefore, for all $i \in [n]$ it holds that 
	\begin{eqnarray*}
		| v_i(O_{+}) - v_i(O_{-}) | \le | v_k(O_{+}) - v_k(O_{-}) | &\le&  | (v_k(O_{+}) - v_k(O_{-})) - (v_k(O'_{+}) - v_k(O'_{-})) | \\&=& | (v_k(O_{+}) - v_k(O'_{+})) - (v_k(O_{-}) - v_k(O'_{-})) | \\&\le& | v_k(O_{+}) - v_k(O'_{+}) | + | v_k(O_{-}) - v_k(O'_{-}) | \\&\le& \frac{\epsilon}{2}+\frac{\epsilon}{2}=\epsilon.
	\end{eqnarray*}
	This means that the partition corresponding to the point ${\bf z}$ is a valid solution to $(n,n,\epsilon)$-\CSH.
\end{proof}

\subsection{From Tucker to Leaf}

\noindent We now proceed to the last step of the proof, establishing the reduction from $(n,T)$-\tucker\ to {\sc Leaf}. {\sc Leaf} is the prototypical problem of the class PPA, based on which the class is actually defined \cite{papadimitriou1994complexity}.

\begin{problem}{\sc Leaf}
	
	\textbf{Input:} A boolean circuit C with $n$ inputs and at most $2n$ outputs, outputting the set $\mathcal{N}(y)$ of (at most two) neighbours of a vertex $y$, such that $|\mathcal{N}(0^n)|=1$.

	\textbf{Output:} A vertex $x$ such that $x \neq 0^n$ and $|\mathcal{N}(x)|=1$.
\end{problem}	
Although it is not necessary for the results of the present section, for completeness, we also define the {\sc End-of-Line} problem, based on which the class PPAD is defined.

\begin{problem}{\sc End-of-Line}

	\textbf{Input:} Two boolean circuits $S$ (for successor) and $P$ (for predecessor) with $n$ inputs and $n$ outputs such that $P(0^n)=0^n \neq S(0^n)$.

\textbf{Output:} A vertex $x$ such that $P(S(x)) \neq x$ or $S(P(x))\neq x \neq 0^n$.
\end{problem}	
We have the following lemma. Note that the mesh size of the triangulation is not a parameter in the statement and the reduction holds for any such size. The lemma actually reduces $(n,T)$-\tucker\ to \textsc{Leaf}, for any fixed centrally symmetric triangulation that is aligned with hemispheres, and in particular any such triangulation of the mesh size needed for Lemma \ref{thm:inPPA} (which clearly holds if $T$ is aligned with hemispheres).\footnote{Strictly speaking, Lemma \ref{lem:tuckerppa} states that all computational problems $(n,T)$-\tucker, parametrized by a fixed triangulation $T$ are in PPA, as long as $T$ is centrally symmetric and aligned with hemispheres.}

\begin{lemma}\label{lem:tuckerppa}
	Let $T$ be a fixed centrally symmetric triangulation that is aligned with hemispheres. Then $(n,T)$-\tucker\ is in PPA. 
\end{lemma}

\begin{proof}
	We reduce $(n,T)$-\tucker\ to the PPA problem {\sc Leaf}, by using the proof by Prescott and Su \cite{prescott2005constructive}. In \cite{prescott2005constructive}, the authors present a constructive proof of Fan's combinatorial lemma on labelings of triangulated spheres. Fan's lemma says that given a symmetric barycentric subdivision of the octahedral subdivision of the $n$-sphere $S^n$ and a labelling from its vertices to $\{ \pm1, \ldots, \pm m\}$, where $m\ge n+1$, such that labels at antipodal vertices sum to 0 and labels at adjacent vertices do not sum to 0, then there are an odd number of $n$-simplices whose labels are of the form
	\begin{ceqn}
	 \[\{ k_0, -k_1, k_2, \ldots, (-1)^n k_n \}, \ \ \text{where} \ \ 1\le k_0 < k_1 < \ldots < k_n \le m.\] 
	 \end{ceqn}
	 Their proof generalizes Fan's lemma in the sense that the result holds for a larger class of triangulations of $S^n$, that is, the lemma holds for any symmetric triangulation of $S^n$ that is aligned with hemispheres. Following their proof, we can start from any place on the sphere and construct a path whose endpoint is a simplex of the highest dimension of the form $\{ k_0, -k_1, k_2, \ldots, (-1)^n k_n \}$. Fan's lemma is a generalization of Tucker's lemma in the sense that if fewer labels are allowed, i.e., $m=n$, then inevitably there will be adjacent vertices labeled by opposite colors.
	We sketch the proof of \cite{prescott2005constructive} and show how it can be converted into a reduction from $(n,T)$-\tucker\ to {\sc Leaf}.
	
	Given a triangulation aligned with hemispheres, we label the vertices of the triangulation as stated in Fan's lemma. The \emph{carrier hemisphere} of a simplex $\sigma$ in $T$ is the minimal $H_d$ or $-H_d$ that contains $\sigma$. A simplex is {\em alternating} if its vertex labels are distinct in magnitude and alternate in sign when arranged in monotone order of magnitude, i.e., the labels have the form 
	\begin{ceqn}
	\[\{ k_0, -k_1, k_2, \ldots, (-1)^n k_n \} \ \ \text{or} \ \ \{ -k_0, k_1, -k_2, \ldots, (-1)^{n+1} k_n \}.
	\] 
	\end{ceqn}
	The {\em sign} of an alternating simplex is the sign of its smallest label.  A simplex is {\em almost-alternating} if it is not alternating, but by deleting one of the vertices, the resulting simplex (a facet) is alternating. The sign of an almost-alternating simplex is defined to be the sign of any of its alternating facets. Thus any almost-alternating simplex must have exactly two facets that are alternating. Call an alternating or almost-alternating simplex {\em agreeable} if the sign of that simplex matches the sign of its carrier hemisphere. 
	
	Now we define a graph $G$. A simplex $\sigma$ carried by $H_d$ is a node of $G$ if it is one of the following: (1) an agreeable alternating $(d-1)$-simplex; (2) an agreeable almost-alternating $d$-simplex; or (3) an alternating $d$-simplex. Two nodes $\sigma$ and $\eta$ are adjacent in $G$ if all of the following hold: (a) one is the facet of the other, (b) $\sigma \cap \eta$ is alternating, and (c) the sign of the carrier hemisphere of $\sigma \cup \eta$ matches the sign of $\sigma \cap \eta$. Prescott and Su show that $G$ is a graph in which every vertex has degree 1 or 2. Furthermore, a vertex has degree 1 if and only if its simplex is carried by $\pm H_d$ or is an $n$-dimensional alternating simplex.
	
	To see how Fan's lemma implies Tucker's lemma and how Prescott and Su's proof can be converted to a reduction to {\sc Leaf}, we now restrict our attention to the case when $m=n$.  Since there are not enough labels for the existence of alternating $n$-simplices, there must exist some agreeable almost-alternating simplices with a complementary edge. We add those simplices as nodes to the graph $G$; it is easy to see that these nodes will be of degree 1. Since the path considered in \cite{prescott2005constructive} can start from any vertex, we can choose any vertex $H_0$ as the given degree 1 node in graph $G$. Following the path, the other degree 1 node in $G$ corresponds to the almost-alternating simplex with a complementary edge. The graph $G$ clearly only contains degree nodes of degree 1 or 2. 
\end{proof}

\section{Conclusion and Future Work} \label{sec:conclusion}

Our work takes an extra step in the direction of capturing the exact complexity of the Consensus-halving problem for all precision parameters. While we believe that the techniques used in \cite{filos2018consensus} can prospectively be used to obtain PPA-hardness of the problem for an inverse-polynomial precision parameter, it seems unlikely that they could applicable when the precision is constant. 
In that sense, our main result is not implied by \cite{filos2018consensus}, neither can it be subsumed by modifications to their reduction, even those involving highly non-trivial alterations. At the same time, our results are useful as the hardness for constant precision established here allowed the authors of \cite{filos2018consensus} to prove PPAD-hardness of the Necklace Splitting problem, for which any hardness results were not previously known.

Going beyond the Consensus-Halving problem in its original definition, it looks interesting to consider further the effects of allowing additional cuts, and the computation of exact or approximate solutions. Concretely, one might wonder how good an approximate Consensus-Halving solution can be computed in polynomial time, given, say, two cuts for each agent or how many cuts are required to produce an approximate solution for a given precision parameter. Finally, we could consider a \emph{query model}, where the agents interact with the protocol by answering value or demand queries and the goal is to find bounds on the number of queries needed to compute approximate solutions. A very similar approach has been taken recently for the related fair-division problem of envy-free cake cutting \cite{branzei2017query} or earlier \cite{deng2012algorithmic} under different assumptions on the agents' valuations.




\bibliographystyle{plainurl}
\bibliography{consensus}

\end{document}